\documentclass[a4paper,11pt]{article}
\usepackage{a4wide}
\usepackage[UKenglish]{babel}
\usepackage[noadjust]{cite}
\usepackage{amsmath}
\usepackage{amssymb}
\usepackage{amsthm} 
\usepackage{bm}
\usepackage{bbm}
\usepackage{mathtools}
\usepackage{mathrsfs}
\usepackage{csquotes}
\usepackage{todonotes}
\usepackage{tikz}
\usepackage{tikz-3dplot}
\usetikzlibrary{perspective}
\usepackage{graphicx,wrapfig}
\usepackage[title]{appendix}
\usepackage{bbm}
\usepackage{xcolor}
\usepackage{blindtext}

\colorlet{minusTwoColour}{blue!200}
\colorlet{minusOneColour}{green!50}
\colorlet{zeroColour}{black!10}
\colorlet{plusOneColour}{yellow!80}
\colorlet{plusTwoColour}{orange!80}
\colorlet{plusThreeColour}{red!200}

\pgfmathsetmacro{\bigShift}{4.5}

\newcommand\squareTensorA[2]{
\begin{tikzpicture}[baseline=#1 ex, local bounding box=C,scale=#2]
\draw[fill=zeroColour,scale=.7] (0,2) rectangle ++(1,1);
\draw[fill=zeroColour,scale=.7] (1,2) rectangle ++(1,1);
\draw[fill=plusOneColour,scale=.7] (2,2) rectangle ++(1,1);
\draw[fill=plusOneColour,scale=.7] (0,1) rectangle ++(1,1);
\draw[fill=zeroColour,scale=.7] (1,1) rectangle ++(1,1);
\draw[fill=minusOneColour,scale=.7] (2,1) rectangle ++(1,1);
\draw[fill=zeroColour,scale=.7] (0,0) rectangle ++(1,1);
\draw[fill=zeroColour,scale=.7] (1,0) rectangle ++(1,1);
\draw[fill=zeroColour,scale=.7] (2,0) rectangle ++(1,1);
\end{tikzpicture}
}

\newcommand\squareTensorB[2]{
\begin{tikzpicture}[baseline=#1 ex, local bounding box=C,scale=#2]
\draw[fill=minusOneColour,scale=.7] (0,2) rectangle ++(1,1);
\draw[fill=zeroColour,scale=.7] (1,2) rectangle ++(1,1);
\draw[fill=plusOneColour,scale=.7] (2,2) rectangle ++(1,1);
\draw[fill=plusTwoColour,scale=.7] (0,1) rectangle ++(1,1);
\draw[fill=zeroColour,scale=.7] (1,1) rectangle ++(1,1);
\draw[fill=minusOneColour,scale=.7] (2,1) rectangle ++(1,1);
\draw[fill=minusOneColour,scale=.7] (0,0) rectangle ++(1,1);
\draw[fill=plusOneColour,scale=.7] (1,0) rectangle ++(1,1);
\draw[fill=zeroColour,scale=.7] (2,0) rectangle ++(1,1);
\end{tikzpicture}
}

\newcommand\rectangularTensorA[2]{
\begin{tikzpicture}[baseline=#1 ex, local bounding box=C,scale=#2]
\draw[fill=zeroColour,scale=.7] (0,2) rectangle ++(1,1);
\draw[fill=zeroColour,scale=.7] (1,2) rectangle ++(1,1);
\draw[fill=plusOneColour,scale=.7] (0,1) rectangle ++(1,1);
\draw[fill=zeroColour,scale=.7] (1,1) rectangle ++(1,1);
\draw[fill=zeroColour,scale=.7] (0,0) rectangle ++(1,1);
\draw[fill=zeroColour,scale=.7] (1,0) rectangle ++(1,1);
\end{tikzpicture}
}

\newcommand\lineTensorA[2]{
\begin{tikzpicture}[baseline=#1 ex, local bounding box=C,scale=#2]
\draw[fill=plusOneColour,scale=.7] (0,2) rectangle ++(1,1);
\draw[fill=minusOneColour,scale=.7] (0,1) rectangle ++(1,1);
\draw[fill=zeroColour,scale=.7] (0,0) rectangle ++(1,1);
\end{tikzpicture}
}

\pgfmathsetmacro{\opacityColouredBlocks}{.8}
\pgfmathsetmacro{\border}{0.1}
\newcommand\mycube[5]{
\draw[line width=\border,fill=#1,opacity=#2,shift={(-#5,#4,-#3)},scale=.7](0,0,0)--++(0,-1,0)--++(0,0,-1)--++(0,1,0)--cycle;
\draw[line width=\border,fill=#1,opacity=#2,shift={(-#5,#4,-#3)},scale=.7](0,0,0)--++(-1,0,0)--++(0,0,-1)--++(1,0,0)--cycle;
\draw[line width=\border,fill=#1,opacity=#2,shift={(-#5,#4,-#3)},scale=.7](0,0,0)--++(-1,0,0)--++(0,-1,0)--++(1,0,0)--cycle;
}

\newcommand\mycubeNoBorder[5]{
\fill[#1,opacity=#2,shift={(-#5,#4,-#3)},scale=.7](0,0,0)--++(0,-1,0)--++(0,0,-1)--++(0,1,0)--cycle;
\fill[#1,opacity=#2,shift={(-#5,#4,-#3)},scale=.7](0,0,0)--++(-1,0,0)--++(0,0,-1)--++(1,0,0)--cycle;
\fill[#1,opacity=#2,shift={(-#5,#4,-#3)},scale=.7](0,0,0)--++(-1,0,0)--++(0,-1,0)--++(1,0,0)--cycle;
}

\newcommand\cubeTensorA[2]{
\begin{tikzpicture}[baseline=#1 ex, local bounding box=C,scale=#2]
\begin{scope}[3d view={110}{15}]
  \foreach \x in {0,...,2}
    \foreach \y in {0,...,2} 
    	\foreach \z in {0,...,2}
    		\mycubeNoBorder{zeroColour}{.3}{\x}{\y}{\z};

\mycube{minusOneColour}{\opacityColouredBlocks}{2}{2}{1};

\mycube{plusOneColour}{\opacityColouredBlocks}{0}{0}{0};
\mycube{minusOneColour}{\opacityColouredBlocks}{1}{0}{0};
\mycube{plusOneColour}{\opacityColouredBlocks}{2}{0}{0};
\mycube{minusOneColour}{\opacityColouredBlocks}{2}{1}{0};
\mycube{plusOneColour}{\opacityColouredBlocks}{2}{2}{0};
\end{scope}
\end{tikzpicture}
}

\newcommand\cubeTensorB[2]{
\begin{tikzpicture}[baseline=#1 ex, local bounding box=C,scale=#2]
\begin{scope}[3d view={110}{15}]
  \foreach \x in {0,...,2}
    \foreach \y in {0,...,2} 
    	\foreach \z in {0,...,2}
    		\mycubeNoBorder{zeroColour}{.3}{\x}{\y}{\z};
    		
\mycube{plusOneColour}{\opacityColouredBlocks}{0}{0}{2};
\mycube{minusOneColour}{\opacityColouredBlocks}{1}{0}{2};
\mycube{plusOneColour}{\opacityColouredBlocks}{2}{0}{2};
\mycube{minusOneColour}{\opacityColouredBlocks}{2}{1}{2};    		
    		
\mycube{plusOneColour}{\opacityColouredBlocks}{2}{2}{1};

\mycube{minusTwoColour}{\opacityColouredBlocks}{0}{0}{0};
\mycube{plusOneColour}{\opacityColouredBlocks}{0}{2}{0};
\mycube{plusThreeColour}{\opacityColouredBlocks}{1}{0}{0};
\mycube{minusOneColour}{\opacityColouredBlocks}{1}{2}{0};
\mycube{minusTwoColour}{\opacityColouredBlocks}{2}{0}{0};
\mycube{plusTwoColour}{\opacityColouredBlocks}{2}{1}{0};
\mycube{minusOneColour}{\opacityColouredBlocks}{2}{2}{0};
\end{scope}
\end{tikzpicture}
}

\newcommand\crystalTensor[2]{
\begin{tikzpicture}[baseline=#1 ex, local bounding box=C,scale=#2]
\begin{scope}[shift={(5,5,5)},3d view={110}{15},scale=.5]
  \foreach \x in {0,...,2}
    \foreach \y in {0,...,2} 
    	\foreach \z in {0,...,2}
    		\mycubeNoBorder{zeroColour}{.3}{\x}{\y}{\z};

\mycube{minusOneColour}{\opacityColouredBlocks}{2}{2}{1};

\mycube{plusOneColour}{\opacityColouredBlocks}{0}{0}{0};
\mycube{minusOneColour}{\opacityColouredBlocks}{1}{0}{0};
\mycube{plusOneColour}{\opacityColouredBlocks}{2}{0}{0};
\mycube{minusOneColour}{\opacityColouredBlocks}{2}{1}{0};
\mycube{plusOneColour}{\opacityColouredBlocks}{2}{2}{0};
\end{scope}

\begin{scope}[shift={(3,3,3)},3d view={110}{15},scale=.75]
  \foreach \x in {0,...,2}
    \foreach \y in {0,...,2} 
    	\foreach \z in {0,...,2}
    		\mycubeNoBorder{zeroColour}{.3}{\x}{\y}{\z};
\end{scope}

\begin{scope}[3d view={110}{15},scale=1]
  \foreach \x in {0,...,2}
    \foreach \y in {0,...,2} 
    	\foreach \z in {0,...,2}
    		\mycubeNoBorder{zeroColour}{.3}{\x}{\y}{\z};
    		
\mycube{plusOneColour}{\opacityColouredBlocks}{0}{0}{2};
\mycube{minusOneColour}{\opacityColouredBlocks}{1}{0}{2};
\mycube{plusOneColour}{\opacityColouredBlocks}{2}{0}{2};
\mycube{minusOneColour}{\opacityColouredBlocks}{2}{1}{2};    		
    		
\mycube{plusOneColour}{\opacityColouredBlocks}{2}{2}{1};

\mycube{minusTwoColour}{\opacityColouredBlocks}{0}{0}{0};
\mycube{plusOneColour}{\opacityColouredBlocks}{0}{2}{0};
\mycube{plusThreeColour}{\opacityColouredBlocks}{1}{0}{0};
\mycube{minusOneColour}{\opacityColouredBlocks}{1}{2}{0};
\mycube{minusTwoColour}{\opacityColouredBlocks}{2}{0}{0};
\mycube{plusTwoColour}{\opacityColouredBlocks}{2}{1}{0};
\mycube{minusOneColour}{\opacityColouredBlocks}{2}{2}{0};
\end{scope}

\end{tikzpicture}
}

\newcommand\finalFigure[2]{
\begin{tikzpicture}[baseline=#1 ex, local bounding box=C,scale=#2,3d view={110}{15}]
\begin{scope}[shift={(-\bigShift,0,0)}]
  \foreach \x in {0,...,2}
    \foreach \y in {0,...,2} 
    	\foreach \z in {0,...,2}
    		\mycubeNoBorder{zeroColour}{.3}{\x}{\y}{\z};

\mycube{plusOneColour}{\opacityColouredBlocks}{0}{0}{2};
\mycube{minusOneColour}{\opacityColouredBlocks}{1}{1}{2};
\mycube{plusOneColour}{\opacityColouredBlocks}{1}{1}{0};

\end{scope}
\begin{scope}[shift={(-\bigShift,0,-\bigShift)}]
  \foreach \x in {0,...,2}
    \foreach \y in {0,...,2} 
    	\foreach \z in {0,...,2}
    		\mycubeNoBorder{zeroColour}{.3}{\x}{\y}{\z};

\mycube{plusOneColour}{\opacityColouredBlocks}{2}{0}{2};
\mycube{minusOneColour}{\opacityColouredBlocks}{2}{1}{2};
\mycube{plusOneColour}{\opacityColouredBlocks}{0}{1}{0};

\end{scope}
\begin{scope}[shift={(-\bigShift,\bigShift,0)}]
  \foreach \x in {0,...,2}
    \foreach \y in {0,...,2} 
    	\foreach \z in {0,...,2}
    		\mycubeNoBorder{zeroColour}{.3}{\x}{\y}{\z};

\mycube{plusOneColour}{\opacityColouredBlocks}{0}{2}{2};
\mycube{minusOneColour}{\opacityColouredBlocks}{1}{2}{2};
\mycube{plusOneColour}{\opacityColouredBlocks}{1}{0}{0};

\end{scope}
\begin{scope}[shift={(-\bigShift,\bigShift,-\bigShift)}]
  \foreach \x in {0,...,2}
    \foreach \y in {0,...,2} 
    	\foreach \z in {0,...,2}
    		\mycubeNoBorder{zeroColour}{.3}{\x}{\y}{\z};

\mycube{plusOneColour}{\opacityColouredBlocks}{0}{0}{0};

\end{scope}

\begin{scope}[shift={(0,0,0)}]
  \foreach \x in {0,...,2}
    \foreach \y in {0,...,2} 
    	\foreach \z in {0,...,2}
    		\mycubeNoBorder{zeroColour}{.3}{\x}{\y}{\z};

\mycube{plusOneColour}{\opacityColouredBlocks}{0}{0}{0};
\end{scope}
\begin{scope}[shift={(0,0,-\bigShift)}]
  \foreach \x in {0,...,2}
    \foreach \y in {0,...,2} 
    	\foreach \z in {0,...,2}
    		\mycubeNoBorder{zeroColour}{.3}{\x}{\y}{\z};
    		
\mycube{plusOneColour}{\opacityColouredBlocks}{0}{1}{1};  
\mycube{minusOneColour}{\opacityColouredBlocks}{2}{1}{1};  
\mycube{plusOneColour}{\opacityColouredBlocks}{2}{0}{0};

\end{scope}
\begin{scope}[shift={(0,\bigShift,0)}]
  \foreach \x in {0,...,2}
    \foreach \y in {0,...,2} 
    	\foreach \z in {0,...,2}
    		\mycubeNoBorder{zeroColour}{.3}{\x}{\y}{\z};

\mycube{plusOneColour}{\opacityColouredBlocks}{1}{0}{1};
\mycube{minusOneColour}{\opacityColouredBlocks}{1}{2}{1};
\mycube{plusOneColour}{\opacityColouredBlocks}{0}{2}{0};

\end{scope}
\begin{scope}[shift={(0,\bigShift,-\bigShift)}]
  \foreach \x in {0,...,2}
    \foreach \y in {0,...,2} 
    	\foreach \z in {0,...,2}
    		\mycubeNoBorder{zeroColour}{.3}{\x}{\y}{\z};

\mycube{plusOneColour}{\opacityColouredBlocks}{0}{0}{1};
\mycube{minusOneColour}{\opacityColouredBlocks}{2}{2}{1};
\mycube{plusOneColour}{\opacityColouredBlocks}{2}{2}{0};

\end{scope}

\end{tikzpicture}
}

\usepackage{hyperref} 
\hypersetup{colorlinks=true,citecolor=blue,linkcolor=blue,urlcolor=blue}

\newcommand{\YES}{\textsc{Yes}}
\newcommand{\NO}{\textsc{No}}
\newcommand{\Test}[2]{
\def\temp{#2}\ifx\temp\empty
  \operatorname{Test}_{#1}
\else
  \operatorname{Test}_{#1}^{#2}
\fi
}

\newcommand{\ignore}[1]{}
\newcommand{\A}{\mathbf{A}}

\newcommand{\GG}{\mathbf{G}}
\newcommand{\HH}{\mathbf{H}}
\newcommand{\K}{\mathbf{K}}

\newcommand{\N}{\mathbb{N}}

\newcommand{\Z}{\mathbb{Z}}

\newcommand{\cC}{\mathcal{C}}

\newcommand{\cT}{\mathcal{T}}

\newcommand{\freeM}{\mathbb{F}_{\Mminion}}

\newcommand{\freeZ}{\mathbb{F}_{\Zaff}}

\newcommand{\Gk}{\GG^\tensor{k}}
\newcommand{\Hk}{\HH^\tensor{k}}

\renewcommand{\vec}[1]{\mathbf{#1}}
\newcommand{\ba}{\vec{a}}
\newcommand{\bb}{\vec{b}}
\newcommand{\bc}{\vec{c}}

\newcommand{\bg}{\vec{g}}
\newcommand{\bh}{\vec{h}}
\newcommand{\bi}{\vec{i}}
\newcommand{\bj}{\vec{j}}
\newcommand{\bn}{\vec{n}}
\newcommand{\bp}{\vec{p}}
\newcommand{\bq}{\vec{q}}
\newcommand{\br}{\vec{r}}
\newcommand{\bs}{\vec{s}}
\newcommand{\bu}{\vec{u}}
\newcommand{\bx}{\vec{x}}
\newcommand{\bv}{\vec{v}}

\newcommand{\bw}{\vec{w}}

\newcommand{\bell}{{\ensuremath{\boldsymbol\ell}}}

\newcommand{\bz}{\vec{z}}
\newcommand{\bepsilon}{{\bm{\epsilon}}}

\newcommand{\balpha}{{\bm{\alpha}}}
\newcommand{\bbeta}{{\bm{\beta}}}

\DeclareMathOperator{\tr}{tr}

\DeclareMathOperator{\AIP}{AIP}

\DeclareMathOperator{\PCSP}{PCSP}
\DeclareMathOperator{\CSP}{CSP}

\DeclareMathOperator{\PCSPs}{PCSPs}
\DeclareMathOperator{\CSPs}{CSPs}

\DeclareMathOperator{\id}{id}

\DeclarePairedDelimiter{\floor}{\lfloor}{\rfloor}

\newcommand{\Mminion}{\ensuremath{{\mathscr{M}}}}

\newcommand{\Zaff}{\ensuremath{{\mathscr{Z}_{\operatorname{aff}}}}}

\newcommand{\bone}{\mathbf{1}}  
 
\newcommand{\tensor}[1]{\textsuperscript{\raisebox{-.5pt}{\normalfont\textcircled{\raisebox{-.1pt}{\tiny #1}}}}}
\newcommand\cont[1]{\mathrel{\overset{\makebox[0pt]{\mbox{\normalfont\tiny\sffamily #1}}}{\ast }}}

\newcommand\ang[1]{{\ensuremath\langle #1\rangle}}

\theoremstyle{plain}
\newtheorem{thm}{Theorem}
\newtheorem*{thm*}{Theorem}
\newtheorem{lem}[thm]{Lemma}
\newtheorem*{lem*}{Lemma}
\newtheorem{prop}[thm]{Proposition}
\newtheorem*{prop*}{Proposition}

\newtheorem*{cor*}{Corollary}

\theoremstyle{definition}
\newtheorem{defn}[thm]{Definition}
\newtheorem{rem}[thm]{Remark}
\newtheorem{example}[thm]{Example}

\newcommand{\lemeq}[1]
{
\ensuremath{\stackrel{\operatorname{L}.\ref{#1}}{\;\;=\;\;}}
}

\newcommand{\defeq}[1]
{
\ensuremath{\stackrel{\operatorname{D}.\ref{#1}}{\;\;=\;\;}}
}
\newcommand{\equationeq}[1]
{
\ensuremath{\stackrel{\eqref{#1}}{\;\;=\;\;}}
}

\begin{document}

\title{
Approximate Graph Colouring and Crystals\thanks{The research leading to these results has received funding from the European Research Council (ERC) under the European Union's Horizon 2020 research and innovation programme (grant agreement No 714532). The paper reflects only the authors' views and not the views of the ERC or the European Commission. The European Union is not liable for any use that may be made of the information contained therein.
This work was also supported by UKRI EP/X024431/1. For the purpose of Open Access, the authors have applied a CC BY public copyright licence to any Author Accepted Manuscript version arising from this submission. All data is provided in full in the results section of this paper.}}

\author{Lorenzo Ciardo\\
University of Oxford, UK\\
\texttt{lorenzo.ciardo@cs.ox.ac.uk} 
\and 
Stanislav {\v{Z}}ivn{\'y}\\
University of Oxford, UK\\
\texttt{standa.zivny@cs.ox.ac.uk}
}

\date{\today}
\maketitle

\begin{abstract}
We show that approximate graph colouring is not solved by any level of the affine integer programming 
($\AIP$) 
hierarchy. To establish the result, we translate the problem of exhibiting a graph fooling a level of the $\AIP$ hierarchy into the problem of constructing a highly symmetric \emph{crystal} tensor. In order to prove the existence of crystals in arbitrary dimension, we provide a combinatorial characterisation for \emph{realisable} systems of tensors; i.e., sets of low-dimensional tensors that can be realised as the projections of a single high-dimensional tensor. 
\end{abstract}

\section{Introduction}
\label{sec_intro}

The \emph{approximate graph colouring} problem (AGC) asks to find a $d$-colouring of a given
$c$-colourable graph, where $3\leq c\leq d$. There is a huge gap in our
understanding of the computational complexity of this problem.
For an $n$-vertex graph and $c=3$, the currently best known polynomial-time
algorithm finds a $d$-colouring of the graph with $d=O(n^{0.19996})$.
It has been long conjectured~\cite{GJ76} that the problem is NP-hard for any fixed constants
$3\leq c\leq d$ even in the decision variant: Given a graph, output $\YES$ if
it is $c$-colourable and output $\NO$ if it is not $d$-colourable.

For $c=d$, the problem becomes the classic $c$-colouring problem, which appeared
on Karp's original list of $21$ NP-complete problems~\cite{Karp72}. The case
$c=3$, $d=4$ was only proved to be NP-hard in 2000~\cite{KhannaLS00}
(and a simpler proof was given 
in~\cite{GK04}); more generally, \cite{KhannaLS00} showed hardness of the case
$d=c+2\floor{c/3}-1$. This was improved to $d=2c-2$ in 2016~\cite{BrakensiekG16}, and recently to $d=2c-1$~\cite{BBKO21}. In particular, this last result implies
hardness of the case $c=3$, $d=5$; the complexity of the case $c=3$, $d=6$ is
still open. Building on~\cite{Khot01,Huang13},
NP-hardness was established for  $d={c\choose\floor{c/2}}-1$
for $c\geq 4$ in~\cite{WZ20}.
NP-hardness of
AGC was established for all constants $3\leq c\leq d$ 
in~\cite{Dinur09:sicomp} under a non-standard variant of
the Unique Games Conjecture, 
in~\cite{GS20:icalp}
under the $d$-to-1 conjecture~\cite{Khot02stoc} for any fixed $d$, and (an even stronger statement of distinguishing 3-colourability from not having an independent set of significant size)
in~\cite{Braverman21:focs} under the rich $2$-to-$1$ conjecture~\cite{Braverman21:itcs}.

AGC is an example of so called \emph{promise constraint satisfaction problem} ($\PCSP$).
For a positive integer ${k}$, a \emph{${k}$-uniform hypergraph} 
$\HH$ consists of a set ${\mathcal{V}}(\HH)$ of elements called \emph{vertices} and a set ${\mathcal{E}}(\HH)\subseteq {\mathcal{V}}(\HH)^{k}$ of tuples of ${k}$ vertices called \emph{hyperedges}.\footnote{Unless otherwise stated, all hypergraphs appearing in this paper are finite, meaning that their vertex set is finite.}
Given two ${k}$-uniform hypergraphs $\GG$ and $\HH$, a map $f:{\mathcal{V}}(\GG)\to {\mathcal{V}}(\HH)$ is a \emph{homomorphism} from $\GG$ to $\HH$ if $f(\bg)\in {\mathcal{E}}(\HH)$ for any $\bg\in {\mathcal{E}}(\GG)$, where $f$ is applied entrywise to the vertices in $\bg$. We shall denote the existence of a homomorphism from $\GG$ to $\HH$ by the expression $\GG\to\HH$.
The $\PCSP$ parameterised by two ${k}$-uniform hypergraphs $\HH$ and $\tilde{\HH}$ such that $\HH\to\tilde{\HH}$, denoted by $\PCSP(\HH,\tilde{\HH})$, is the following computational problem: Given a ${k}$-uniform hypergraph $\GG$ as input, answer $\YES$ if $\GG\to\HH$ and $\NO$ if $\GG\not\to\tilde{\HH}$. The requirement $\HH\to\tilde{\HH}$ implies that the two cases cannot happen simultaneously, as homomorphisms compose; the \emph{promise} is that one of the two cases always happens.\footnote{It is customary to study $\PCSPs$ on more general objects known as \emph{relational structures}, which consist of a collection of  relations of arbitrary arities on a vertex set.}
A $2$-uniform hypergraph is a \emph{digraph}.
Moreover, a $p$-colouring of a digraph $\GG$ is precisely a homomorphism from $\GG$ to the \emph{clique} $\K_p$ -- i.e., the digraph on vertex set $\{1,\dots,p\}$ such that any pair of distinct vertices is a (directed) edge. Therefore, AGC is $\PCSP(\K_c,\K_d)$.

By letting $\HH=\tilde{\HH}$ in the definition of a $\PCSP$, one obtains the standard (non-promise) \emph{constraint satisfaction problem} ($\CSP$). $\PCSPs$ were introduced in~\cite{AGH17,BG21:sicomp}
as a general framework for studying approximability of perfectly satisfiable $\CSPs$ and have
emerged as a new direction in constraint satisfaction that requires different
techniques than $\CSPs$. Recent works on $\PCSPs$ include
those using analytical methods~\cite{Bhangale21:stoc,Braverman21:itcs,BGS21,Bhangale22:stoc}
and those building on algebraic methods~\cite{BG19,bgwz20,WZ20,GS20:icalp,
AB21,Barto21:stacs,BWZ21,Butti21:mfcs,Barto22:soda,CZ22:sicomp-clap,NZ22:icalp}
developed in~\cite{BBKO21}. However, most basic questions are still left open, including applicability of different types of algorithms.
Remarkably, most algorithmic techniques in constraint satisfaction can be broadly classified into two distinct classes: Algorithms enforcing some type of local consistency, and algorithms
related to (generalisations of) linear equations.

The first class of algorithms is based on the following idea: Rather than directly checking for the existence of a global map between 
$\GG$ and $\HH$ satisfying  constraints (i.e., 
a homomorphism), which may not be doable in polynomial time, one 
tries to draw an \emph{atlas} of \emph{charts} covering each region of the instance $\GG$. The charts are partial homomorphisms, i.e., homomorphisms from a substructure of $\GG$ to $\HH$; the atlas must have the property that the maps are consistent, i.e., whenever two regions overlap, there exist charts of the regions that agree on the intersection. The \emph{bounded width} (or \emph{local-consistency checking}) algorithm outputs $\YES$ if and only if such an atlas exists -- which can be checked in polynomial time provided that the size of each chart is bounded~\cite{Feder98:monotone}. More powerful versions of this technique require that the charts of each region should be sampled according to some probability distribution. In this case, the consistency requirement of the atlas is stronger, as it asks that, whenever two regions overlap, 
the probability distribution over the charts of the intersection should be exactly the \emph{marginal} of the distributions over the charts of the two regions. Concretely, checking for the existence of such a ``random atlas'' amounts to solving a linear program, and results in the so-called Sherali-Adams LP hierarchy~\cite{Sherali1990}, which is provably more powerful than bounded width~\cite{Atserias22:soda}. Treating probabilities as vectors satisfying certain orthogonality requirements gives an even stronger
algorithm based on semidefinite programming, known as the sum-of-squares or Lasserre SDP hierarchy~\cite{shor1987class,parrilo2000structured,Lasserre02}.
In general, the existence of a (random) atlas is not sufficient to deduce that a \emph{planisphere} (i.e., a global map satisfying all constraints) exists.
In fact, if P$\neq$NP,
we do not expect polynomial-time algorithms to solve NP-hard problems. Thus, a well-established line
of work has sought to prove lower bounds on the efficacy of these consistency algorithms; see~\cite{Arora06:toc,Braun15:stoc,Chan16:jacm-lp,Kothari17:stoc,Ghosh18:toc} for lower bounds on LPs arising from lift-and-project hierarchies such as that of Sherali-Adams, and~\cite{Tulsiani09:stoc,Lee15:stoc,Chan15:jacm} for lower bounds on SDPs.

Any $\PCSP$ can be formulated as a system of linear equations over $\{0,1\}$.
The second class of algorithms essentially consists in solving the equations
using (some variant or a generalisation of) Gaussian elimination. This requires relaxing the problem by admitting a larger range for the variables in the equations (as, in general, the system cannot be efficiently solved over $\{0,1\}$).
In particular, 
it is possible to solve the system in polynomial time over $\Z$
(\cite{kannan1979}, cf.~also~\cite{BG19}) -- which results in the \emph{affine integer programming} ($\AIP$) relaxation, also known as \emph{linear  Diophantine equations},\footnote{A hierarchy based on $\AIP$ (with additional local-consistency conditions) was considered in~\cite{Berkholz17:soda}, where a lower bound on its power was shown for the graph isomorphism problem.}
that we consider in this work. The ``base level'' of this algorithm 
was studied in~\cite{BG19,bgwz20} in the context of PCSPs, and its power was characterised algebraically in~\cite{BBKO21}.
This algorithmic technique is substantially different from the first class of algorithms: Instead of looking for an atlas of charts faithfully describing regions of the world -- i.e., a system of \emph{local assignments satisfying the constraints} -- the algorithms of the second class aim to draw a possibly imprecise planisphere -- i.e., a \emph{global assignment satisfying a relaxed version of the constraints}.
In the context of $\CSPs$, the elusive interaction between consistency-checking
and methods based on (generalisations of) Gaussian-elimination was the major obstacle to the proof of the Feder-Vardi dichotomy conjecture~\cite{Feder98:monotone}, 
that was finally settled in~\cite{Bulatov17:focs} and, independently, in~\cite{Zhuk20:jacm}.

If, as conjectured, AGC is an NP-hard problem and P$\neq$ NP, neither of the two algorithmic techniques should be able to solve it.
In a striking sequence of
works~\cite{Khot17:stoc-independent,Dinur18:stoc-non-optimality,Dinur18:stoc-towards,Khot18:focs-pseudorandom},
the 2-to-2 conjecture of Khot~\cite{Khot02stoc} (with imperfect completeness)
was resolved. As detailed in~\cite{Khot18:focs-pseudorandom}, this implies
(together with~\cite{GS20:icalp}) that polynomially many levels of the sum-of-squares hierarchy do not solve AGC, which implies the same result for polynomially many levels of the weaker Sherali-Adams and bounded width algorithms.
Recent work~\cite{Atserias22:soda} established that even sublinear
levels of bounded width do not solve AGC.

\paragraph{Contributions}
In this paper, we focus on the second class of algorithms and show that no level of the affine integer programming hierarchy solves AGC.
Recently, \cite{CZ22minions} described a linear-algebraic characterisation of the algorithm in terms of 
a geometric construction called \emph{tensorisation}.
Using this characterisation as a black box, we translate the problem of finding an instance of AGC fooling the algorithm into the problem of finding a tensor with many symmetries, which we call a \emph{crystal}. Our main technical contribution is the construction of crystals. More precisely, we prove the following result: Given a collection of low-dimensional tensors (``pictures'') satisfying a compatibility requirement, it is possible to build a high-dimensional tensor 
such that by projecting it onto low-dimensional hyperplanes one recovers the pictures.
Variants of this problem have appeared in the literature in combinatorial matrix theory. In particular, the problem of constructing a matrix (i.e., a $2$-dimensional tensor) having prescribed row-sum and column-sum vectors (i.e., $1$-dimensional projections) has been studied for different classes of matrices, such as nonnegative integer matrices~\cite{brualdi1991combinatorial}, $(0,1)$ matrices~\cite{da20090,ryser1957combinatorial}, alternating-sign matrices~\cite{MR1392498}, and sign-restricted matrices~\cite{brualdi2021sign}, see also the survey~\cite{MR2890890}. For example, the Gale-Ryser theorem~\cite{ryser1957combinatorial}
provides a characterisation, based on the notion of majorisation, of the pairs of vectors $\br,\bc$ for which there exists a $(0,1)$ matrix whose row-sum and column-sum vectors are $\br$ and $\bc$, respectively. 
In a similar fashion, we not only show that a tensor having prescribed low-dimensional projections exists, but we also prove that a natural necessary combinatorial condition is in fact also sufficient for a system of low-dimensional ``picture'' tensors in order to be the set of projections of a high-dimensional 
tensor. 

We point out that our proof is constructive, as it allows to explicitly find a tensor with the desired characteristics. 
As far as we know, the problem of reconstructing a tensor from low-dimensional projections has 
hitherto 
only been studied for matrices (but cf.~\cite{MR3759214}, where a related problem is investigated in three dimensions in the restricted setting of alternating-sign $3$-dimensional tensors). However, in order to rule out affine integer programming as an algorithm to solve AGC for all numbers of colours, we need to build crystals of arbitrarily high dimension 
and hence approach the reconstruction problem for arbitrarily high-dimensional tensors. 
In addition to its direct application
to the non-solvability of AGC, 
we believe that our result might be of independent interest to the  linear algebra and tensor theory communities.

\section{Overview}
\label{sec_overview}

Let $k\geq 2$ be an integer.
Given a set $V$, we define ${V\choose\leq k}=\{S\subseteq V:1\leq|S|\leq k\}$. 
Let $\GG$ and $\HH$ be two digraphs. We introduce a variable $\lambda_S(f)$ for every $S\in{{\mathcal{V}}(\GG)\choose\leq k}$
 and every function $f:S\to {\mathcal{V}}(\HH)$, and a variable $\lambda_\bg(f)$ for every $\bg=(g_1,g_2)\in {\mathcal{E}}(\GG)$ and every $f:\{g_1,g_2\}\to {\mathcal{V}}(\HH)$.
The $k$-th level of the $\AIP$ hierarchy is given by the following constraints:
\[
\begin{array}{lll}
\mbox{(AIP1)} & \displaystyle\sum_{f:S\to {\mathcal{V}}(\HH)}\lambda_S(f)=1 & S\in{{\mathcal{V}}(\GG)\choose\leq k}\\
\mbox{(AIP2)} & \displaystyle \lambda_R(f)=\sum_{\tilde{f}:S\to {\mathcal{V}}(\HH),\;\tilde{f}|_{R}=f}\lambda_S(\tilde{f}) & 
\emptyset\neq R\subseteq S\in  {{\mathcal{V}}(\GG)\choose\leq k}, f:R\to {\mathcal{V}}(\HH) 
\\
\mbox{(AIP3)} & \displaystyle \lambda_R(f)=\sum_{\tilde{f}:\{g_1,g_2\}\to {\mathcal{V}}(\HH),\;\tilde{f}|_{R}=f}\lambda_{\bg}(\tilde{f})
&
\bg\in {\mathcal{E}}(\GG), \emptyset\neq R\subseteq \{g_1,g_2\}, f:R\to {\mathcal{V}}(\HH) 
\\
\mbox{(AIP4)} & \displaystyle \lambda_{\bg}(f)=0 & 
\bg\in {\mathcal{E}}(\GG), f:\{g_1,g_2\}\to {\mathcal{V}}(\HH)
\mbox{ with } f(\bg)\not\in {\mathcal{E}}(\HH).
\end{array}
\]
We say that $\AIP^k(\GG,\HH)=\YES$ if the system above admits a solution such that all variables take integer values. 
For a fixed $k$, this can be checked in polynomial time in the number of vertices of the input digraph $\GG$ by solving a polynomial-sized system of linear equations over the integers~\cite{kannan1979}.
(For the ``base level'' of the hierarchy $k=1$, cf.~Appendix~\ref{sec:K2}.)

Let $\tilde{\HH}$ be a digraph such that $\HH\to\tilde{\HH}$.
One easily checks that $\AIP^k(\GG,\HH)=\YES$ if $\GG\to\HH$;
we say that the $k$-th level of $\AIP$ \emph{solves} $\PCSP(\HH,\tilde\HH)$ if $\GG\to\tilde{\HH}$ whenever $\AIP^k(\GG,\HH)=\YES$.
Clearly, if $\AIP^k(\GG,\HH)=\YES$ then $\AIP^{k'}(\GG,\HH)=\YES$ for any level $k'$ \emph{lower} than $k$.
It follows that if some level of the hierarchy solves $\PCSP(\HH,\tilde{\HH})$ then any \emph{higher} level of the hierarchy also solves it.
It is worth noticing that the $\AIP$ hierarchy does not enforce consistency, in the sense that it is possible that a partial assignment is given nonzero weight without being a partial homomorphism. This is in sharp contrast to the ``consistency-enforcing'' algorithms mentioned in the Introduction, such as the bounded-width, Sherali-Adams LP, and Lasserre SDP hierarchies.
We now state the first main result of this work.

\begin{thm}
\label{thm_AIPk_no_solves_AGC}
No level of the $\AIP$ hierarchy solves  approximate graph colouring; i.e., for any fixed $3\leq c\leq d$, there is no $k$ such that the $k$-th level of $\AIP$ solves $\PCSP(\K_c,\K_d)$.  
\end{thm}

\subsection{Affine integer programming and tensors}
\label{subsec_AIP_and_tensors}

In order to prove Theorem~\ref{thm_AIPk_no_solves_AGC}, we need to find instances of AGC that fool the $\AIP$ hierarchy. Rather than working with the hierarchy itself, we shall lift the analysis to a tensor-theoretic framework. Next, we define some terminology on tensors that will be used throughout the paper.

Given $n$ in the set $\N$ of positive integer numbers, we let $[n]=\{1,\dots,n\}$. We also let $[0]=\emptyset$.
Given a tuple $\bn=(n_1,\dots,n_q)\in\N^q$ for some $q\in\N$, we let $[\bn]=[n_1]\times\dots\times [n_q]$.
Given a tuple $\bb=(b_1,\dots,b_q)\in [\bn]$ and a tuple $\bi=(i_1,\dots,i_p)\in
[q]^p$ for $p,q\in\N$,
the \emph{projection} of $\bb$ onto $\bi$ is the tuple $\bb_\bi= (b_{i_1},\dots,b_{i_p})$. Notice that $\bb_\bi\in [\bn_\bi]$. For $\tilde\bn\in\N^p$, the \emph{concatenation} of two tuples $\bb=(b_1,\dots,b_q)\in [\bn]$ and $\bc=(c_1,\dots,c_p)\in [\tilde{\bn}]$ is the tuple $(\bb,\bc)=(b_1,\dots,b_q,c_1,\dots,c_p)$. Notice that $(\bb,\bc)\in [(\bn,\tilde{\bn})]$. 
It will be handy to extend the notation above to include tuples of length zero.
For any set $S$, we define $S^0= \{\bepsilon\}$, where $\bepsilon$ denotes the empty tuple. For any tuple $\bx$, we let $\bx_\bepsilon=\bepsilon$ and $(\bx,\bepsilon)=(\bepsilon,\bx)=\bx$. We also define $[\bepsilon]= \{\bepsilon\}$. 
For $n\in\N$, define the tuple $\ang{n}=(1,\dots,n)$. Also, let $\ang{0}=\bepsilon$. Given a tuple $\bx$, $\#(\bx)$ is the cardinality of the set of elements appearing in $\bx$.

Let $\N_0=\N\cup\{0\}$.
Take a set $S$, an integer $q\in\N_0$, and a tuple $\bn\in\N^q$. We denote by $\cT^\bn(S)$ the set of \emph{tensors} on $q$ \emph{modes} of sizes $n_1,\dots,n_q$ whose entries are in $S$; formally, $\cT^\bn(S)$ is the set of functions from $[\bn]$ to $S$.
We sometimes denote a tensor in $\cT^\bn(S)$ by $T=(t_\bi)_{\bi\in [\bn]}$, where $t_\bi\in S$ is the image of $\bi$ under $T$.
For example, $\cT^{n}(S)$ and $\cT^{(m,n)}(S)$ are the sets of $n$-vectors and $m\times n$ matrices, respectively, having entries in $S$. Notice that $\cT^\bepsilon(S)$ is the set of functions from $[\bepsilon]=\{\bepsilon\}$ to $S$, which we identify with $S$.
We will often consider \emph{cubical} tensors, all of whose modes have equal size; i.e., tensors in the set $\cT^{n\cdot\bone_q}$ for some $n\in \N$, where $\bone_q$ is the all-one tuple of length $q$.

We shall usually consider tensors having entries in the ring of integers $\Z$.
For $k,\ell,m\in\N_0$, take $\bn\in\N^k$, $\bp\in\N^\ell$, and $\bq\in\N^m$. 
The \emph{contraction} of two tensors $T=(t_\bi)_{\bi\in[(\bn,\bp)]}\in\cT^{(\bn,\bp)}(\Z)$ and $\tilde T=(\tilde t_\bi)_{\bi\in[(\bp,\bq)]}\in\cT^{(\bp,\bq)}(\Z)$, denoted by  
$T\overset{\mathrm{\ell}}{\ast}\tilde{T}$,
is the tensor in $\cT^{(\bn,\bq)}(\Z)$ such that, for $\bi\in [\bn]$ and $\bj\in[\bq]$, the $(\bi,\bj)$-th entry of $T\overset{\mathrm{\ell}}{\ast}\tilde{T}$ is given by 
$
\sum_{\bz\in [\bp]}t_{(\bi,\bz)}\tilde{t}_{(\bz,\bj)}
$.
If at least one of $k$ and $m$ equals zero -- i.e., if we are contracting over all modes of $T$ or $\tilde{T}$, we write $T\ast\tilde{T}$ for $T\overset{\mathrm{\ell}}{\ast}\tilde{T}$, to increase readability. 
It is not hard to see that tensor contraction is associative, in the sense that $(T\overset{\mathrm{\ell}}{\ast}\tilde T)\overset{\mathrm{m}}{\ast}\hat{T}=T\overset{\mathrm{\ell}}{\ast}(\tilde T\overset{\mathrm{m}}{\ast}\hat{T})$ for any $\hat{T}\in\cT^{(\bq,\textbf{r})}(\Z)$, where $\textbf{r}\in\N^n$ for some $n\in\N_0$. On the other hand, the order of operations matters for the ``$\ast$'' operator. For example, if $\overline{T}\in\cT^{\bn}(\Z)$, the expression $(\overline{T}\ast T)\ast \tilde T$ is well defined but the expression $\overline T\ast (T\ast \tilde T)$ is not, in general. For this reason, we define ``$\ast$'' to be left-associative; i.e., 
$T_1\ast T_2\ast T_3$ means $(T_1\ast T_2)\ast T_3$.
The next example shows that contraction generalises various linear-algebraic products.
\begin{example}
For $m,n,p\in\N$, consider the tensors $c\in\cT^{\bepsilon}(\Z)=\Z$, $\bu,\bv\in\cT^{m}(\Z)$, $\bw\in\cT^{n}(\Z)$, $M,N\in\cT^{(m,n)}(\Z)$, and $Q\in\cT^{(n,p)}(\Z)$. The following products can be seen as examples of contraction:
$c\cont{0}\bu=c\ast\bu=c\bu$, $c\cont{0}M=c\ast M=cM$ (multiplication times scalar);
$\bu\cont{1}\bv=\bu\ast\bv=\bu^T\bv$ (inner product);
$\bu\cont{0}\bw=\bu\bw^T$ (outer product);
$M\cont{1}Q=MQ$ (standard matrix product);
$M\cont{2}N=\tr(M^TN)$ (Frobenius inner product).
\end{example}

Let $q\in\N_0$ and $\bn\in\N^q$.
Given $\bi\in[\bn]$, we denote by $E_\bi$ the \emph{$\bi$-th standard unit tensor}; i.e., the tensor in
$\cT^\bn(\Z)$ all of whose entries are $0$, except the
$\bi$-th entry that is $1$. Observe that, for any $T\in \cT^\bn(\Z)$, we may express the $\bi$-th entry of $T$ as
$E_\bi\ast T$. 
If $q=1$, $n\in\N$, and $i\in [n]$, notice that $E_i$ is the $i$-th standard unit vector of length $n$.
Let $\bi\in [q]^p$ for some $p\in\N_0$. We associate with $\bn$ and $\bi$ the tensor $\Pi^\bn_\bi\in\cT^{(\bn_\bi,\bn)}(\Z)$
 defined by 
\begin{align}
\label{eqn_projection_tensor}
E_\ba\ast \Pi^\bn_\bi\ast E_{\bb}=
\left\{
\begin{array}{ll}
1 & \mbox{if }\bb_\bi=\ba\\
0 & \mbox{otherwise}
\end{array}
\right.
\hspace{1cm}
\mbox{for each } 
\begin{array}{lll}
\ba\in [\bn_\bi],\\
\bb\in [\bn].
\end{array}
\end{align}
We will need a few technical lemmas on the tensors defined above,\footnote{The expression ``$x \stackrel{\operatorname{L}.\bullet}{\;=\;} y$'' shall mean ``$x=y$ by Lemma~$\bullet$''. Similarly, ``$x \stackrel{\operatorname{D}.\bullet}{\;=\;} y$'' and ``$x \stackrel{(\bullet)}{\;=\;} y$'' shall mean ``$x=y$ by Definition~$\bullet$'' and ``$x=y$ by equation~$(\bullet)$'', respectively.}
whose proofs can be found
in Section~\ref{sec:proofs}. The first concerns the ``limit case'' of the empty tuple $\bepsilon$.

\begin{lem}
\label{Pi_epsilon_all_one}
$E_\bepsilon=1$. Moreover,
given $q\in\N_0$ and $\bn\in\N^q$, $\Pi^\bn_\bepsilon$ is the all-one tensor in $\cT^\bn(\Z)$.
\end{lem}

The following is a simple description of the entries of $\Pi^\bn_\bi$.

\begin{lem}
\label{lem_Pi_basic}
Given $p,q\in \N_0$, $\bn\in\N^q$, $\bi\in [q]^p$, and $\ba\in [\bn_\bi]$, we have
$
E_\ba\ast \Pi^\bn_\bi=\sum_{\bb\in [\bn],\;\bb_\bi=\ba}E_{\bb}.
$
\end{lem}
\noindent
The assignment $\bi\mapsto\Pi^\bn_\bi$
creates a correspondence between 
tuples and tensors. More specifically, Lemma~\ref{lem_proj_contr} shows that, under this assignment, 
the operation of tuple projection 
is translated 
into the operation of tensor contraction, while Lemma~\ref{lem_identity_Pi} shows that the tuple $\ang{q}$, that acts by projection as the identity on the set of tuples of appropriate length, corresponds to a tensor that acts by contraction as the identity on the space of tensors of appropriate size.
\begin{lem}
\label{lem_proj_contr}
Let $m,p,q\in\N_0$, and consider two tuples $\bi\in [q]^p$ and $\bj\in [p]^m$. Then, for any $\bn\in\N^q$,
$
\Pi^\bn_{\bi_\bj}=\Pi^{\bn_\bi}_\bj \cont{p} \Pi^\bn_\bi.
$
\end{lem}
\begin{lem}
\label{lem_identity_Pi}
Let $q,q'\in\N_0$, $\bn\in\N^q$, $\bn'\in\N^{q'}$, and $T\in\cT^{(\bn,\bn')}(\Z)$. Then
$
\Pi^\bn_{\ang{q}}\cont{q} T=T.
$
\end{lem}

In~\cite{CZ22minions}, the $\AIP$ hierarchy was characterised algebraically by using a multilinear construction. We state the characterisation in Theorem~\ref{prop_acceptance_AIPk} below, after introducing the necessary terminology. 
\begin{defn}[\cite{BBKO21,bgwz20}]
\label{defn_minion_free_structure}
A \emph{minion} $\Mminion$ is the disjoint union of nonempty sets $\Mminion^{(p)}$ for $p\in\N$ equipped with operations $(\cdot)_{/\pi}:\Mminion^{(p)}\to\Mminion^{(q)}$ for all $\pi:[p]\to [q]$ that satisfy, for any $p,q,r\in\N$, $\pi:[p]\to [q]$, $\rho:[q]\to [r]$, $M\in\Mminion^{(p)}$, the requirements $(i)$ $(M_{/\pi})_{/\rho}=M_{/\rho\circ\pi}$, and $(ii)$ $M_{/\id}=M$.

Let $\HH$ be a $k$-uniform hypergraph having $n$ vertices and $m$ hyperedges.
The \emph{free hypergraph} $\freeM(\HH)$ of a minion $\Mminion$ generated by $\HH$ is the (potentially infinite) $k$-uniform hypergraph on the vertex set ${\mathcal{V}}(\freeM(\HH))=\Mminion^{(n)}$ whose hyperedges are defined as follows: Given $M_1,\dots,M_k\in \Mminion^{(n)}$, the tuple $(M_1,\dots,M_k)$ belongs to ${\mathcal{E}}(\freeM(\HH))$ if and only if there exists some $Q\in\Mminion^{(m)}$ such that $M_i=Q_{/\pi_i}$ for any $i\in[k]$, where $\pi_i:{\mathcal{E}}(\HH)\to {\mathcal{V}}(\HH)$ maps a hyperedge $\bh$ to its $i$-th entry $h_i$.
\end{defn}
\begin{example}[\cite{BBKO21}]
\label{example_Zaff}
For any $p\in\N$, let $\Zaff^{(p)}$ be the set of integer vectors of length $p$ whose entries sum up to one. Given $\pi:[p]\to [q]$ and $\bv\in\Zaff^{(p)}$, let $\bv_{/\pi}$ be the $q$-vector whose $j$-th entry is $\sum_{\ell\in\pi^{-1}(j)}v_\ell$ for each $j\in [q]$. One easily shows that $\Zaff=\bigcup_{p\in\N}\Zaff^{(p)}$ is a minion. 
\end{example}
\begin{defn}[\cite{CZ22minions}]
\label{defn_tensorisation}
Given $k\in\N$, the \emph{$k$-th tensor power}\footnote{The expression ``tensor product of digraphs'' is sometimes used in the literature to indicate the direct or categorical product of digraphs. The tensor power used here is unrelated to that notion -- in particular, as it is clear from Definition~\ref{defn_tensorisation}, the $k$-th tensor power of a digraph is not a digraph for $k>1$.} of a digraph $\HH$ is the $2^k$-uniform hypergraph $\Hk$ having vertex set ${\mathcal{V}}(\Hk)={\mathcal{V}}(\HH)^k$ and hyperedge set ${\mathcal{E}}(\Hk)=\{\bh^\tensor{k}:\bh\in {\mathcal{E}}(\HH)\}$, where, for $\bh\in {\mathcal{E}}(\HH)$, $\bh^\tensor{k}$ is the tensor\footnote{In particular, the number of hyperedges in $\Hk$ is equal to the number of edges in $\HH$.} in $\cT^{2\cdot\bone_k}({\mathcal{V}}(\HH)^k)$ whose $\bi$-th entry is $\bh_\bi$ for every $\bi\in [2]^k$.
\end{defn}
\begin{example}
\label{ex_free_struc_tensorised_digraph}
Let us describe the free hypergraph of $\Zaff$ generated by $\Hk$, where $\HH$ is a digraph on $n$ vertices.
$\freeZ(\Hk)$ 
is a (potentially infinite) $2^k$-uniform hypergraph whose vertex set is $\Zaff^{(n^k)}$, which we identify with the set of (cubical) tensors in $\cT^{n\cdot\bone_k}(\Z)$ whose entries sum up to one. 
Each hyperedge of $\freeZ(\Hk)$ 
consists of $2^k$ vertices, i.e., $2^k$ elements of $\Zaff^{(n^k)}$. It is convenient to visualise it as a block tensor $T$ belonging to $\cT^{2\cdot\bone_k}(\cT^{n\cdot\bone_k}(\Z))=\cT^{2n\cdot\bone_k}(\Z)$. Using Definition~\ref{defn_minion_free_structure}, we see that $T\in\mathcal{E}(\freeZ(\Hk))$ if and only if there exists some $Q\in\Zaff^{(|\mathcal{E}(\Hk)|)}=\Zaff^{(|\mathcal{E}(\HH)|)}$ such that, for any $\bi\in [2]^k$, the $\bi$-th block of $T$ is equal to $Q_{/\pi_\bi}$, where $\pi_\bi:\mathcal{E}(\HH)\to\mathcal{V}(\HH)^k$ maps $\bh\in\mathcal{E}(\HH)$ to $\bh_\bi$. It only remains to describe the entries of $Q_{/\pi_\bi}$. According to Example~\ref{example_Zaff}, given any $\bh\in \mathcal{V}(\HH)^k$, the $\bh$-th entry of $Q_{/\pi_\bi}$ is given by
\begin{align}
\label{eqn_1312_0407}
E_{\bh}\ast Q_{/\pi_\bi}
=
\sum_{\bell\in\pi_\bi^{-1}(\bh)}E_\bell\ast Q
=
\sum_{\substack{\bell\in\mathcal{E}(\HH)\\\bell_\bi=\bh}}E_\bell\ast Q.
\end{align} 
\end{example}

\noindent The following result characterises acceptance for the $\AIP$ hierarchy.\footnote{The result in~\cite{CZ22minions} is proved for arbitrary relational structures; 
for the purpose of this work, the less general version concerning digraphs is enough. Moreover, the definition of the $\AIP$ hierarchy and the other hierarchies characterised in~\cite{CZ22minions} is formally different from the definition used here, in that it requires preprocessing the $\PCSP$ template and instance by ``$k$-enhancing'' them, i.e., adding dummy constraints on $k$-tuples of variables. As proved in~\cite[Section A.1]{CZ22minions}, that definition is equivalent to the more standard hierarchy definition used in~\cite{Butti21:mfcs}, which we follow in this work.
} 

\begin{thm}[\cite{CZ22minions}]
\label{prop_acceptance_AIPk}
Let $\GG,\HH$ be two digraphs and let $k\geq 2$. Then $\AIP^k(\GG,\HH)=\YES$ if and only if there exists a homomorphism $\xi:\Gk\to\freeZ(\Hk)$ such that $\xi(\bg_\bi)=\Pi^{n\cdot\bone_k}_\bi\ast \xi(\bg)$
for any  $\bg\in {\mathcal{V}}(\GG)^k, \bi\in [k]^k$.
\end{thm}

\subsection{The quest for crystals}
\label{sec:quest}
Theorem~\ref{thm_AIPk_no_solves_AGC} is established by proving the existence of certain highly symmetric tensors (Theorem~\ref{thm_crystals_exist}, our second main result) and using them to fool the $\AIP$ hierarchy (Proposition~\ref{prop_AIPk_acceptance_graph_colouring}). 
The tensors we will build enjoy the remarkable property of looking identical when observed from any angle, which is why we shall refer to them as to \emph{crystals}.

Given $p,q\in\N$, we let $[q]^p_\rightarrow$ denote the set of increasing tuples in $[q]^p$; i.e., $[q]^p_\rightarrow=\{(i_1,\dots,i_p)\in [q]^p \mbox{ s.t. } i_1<i_2<\dots<i_p\}$. Moreover, we let $[q]^0_\to=\{\bepsilon\}$ for any $q\in\N$.
Observe that $[q]^p_\rightarrow\neq\emptyset$ if and only if $p\leq q$.

\begin{defn}
\label{defn_crystals}
For $q,n\in\N$, let $M$ be an $n\times n$ integer matrix. 
A tensor $C\in\cT^{n\cdot\bone_q}(\Z)$ is a \emph{$q$-dimensional $M$-crystal} if $\Pi^{n\cdot\bone_q}_\bi\ast C = M$ for each $\bi\in [q]^2_{\rightarrow}$. 
\end{defn}
\begin{rem}
\label{rem_1204_0307}
For $\bn\in\N^q$ and $\bi\in [q]^p$, the tensor $\Pi^\bn_\bi$ introduced in Section~\ref{subsec_AIP_and_tensors} should be understood as a projection operator, that projects a given tensor $T$ living in $\cT^\bn(\Z)$ onto a new system of modes -- namely, $\bn_\bi$. As an example, we have seen (cf.~Lemma~\ref{lem_identity_Pi}) that, if $\bi$ is the identity tuple (i.e., the tuple $\ang{q}$), contracting by $\Pi^\bn_\bi$ leaves $T$ unaffected. More in general, if $\bi$ is a permutation (i.e., $\#(\bi)=p=q$), $\Pi^\bn_\bi$ simply rotates the tensor by rearranging its modes. For instance, for $p=q=2$, $\Pi^\bn_{(1,2)}$ is the identity operator, while $\Pi^\bn_{(2,1)}$ is the transpose operator. Indeed, letting $\bn=(n_1,n_2)\in\N^2$ and considering an $n_1\times n_2$ matrix $M$, $\Pi^\bn_{(1,2)}\ast M=M$ and $\Pi^\bn_{(2,1)}\ast M=M^T$. If $p\leq q$, as it is the case for Definition~\ref{defn_crystals}, $\Pi^\bn_\bi$ projects a tensor $T$ having $q$ modes onto a smaller, $p$-dimensional space. In other words, $\Pi^\bn_\bi\ast T$ is a ``$p$-dimensional picture'' of $T$.
\end{rem}

\begin{thm}
\label{thm_crystals_exist}
Let $q,n\in\N$, and let $M$ be an $n\times n$ integer matrix satisfying $M\bone_n=M^T\bone_n$. Then there exists a $q$-dimensional $M$-crystal. 
\end{thm}

Our approach to prove Theorem~\ref{thm_crystals_exist} will be to show something
slightly more general: Given a collection $\cC$ of pictures that is
\emph{realistic} -- i.e., such that each pair of pictures is ``locally compatible'' with each other -- one can always produce a tensor $C$ such that photographing $C$ from all angles results in the pictures in $\cC$. After establishing this result (Proposition~\ref{prop_realistic_equals_truthful}), Theorem~\ref{thm_crystals_exist} will easily follow, by letting all pictures be the same matrix $M$. We note that, even if the pictures in the definition of a crystal are two-dimensional objects (matrices), the results we shall prove are more conveniently phrased in terms of arbitrary-dimensional pictures.    

\begin{defn}
\label{defn_album_of_pictures}
For $p,q\in\N$ and $\bn\in\N^q$, a \emph{$(p,\bn)$-album of pictures} is a set $\cC=\{C_\bi\}_{\bi\in [q]^p_\rightarrow}$ such that $C_\bi\in\cT^{\bn_\bi}(\Z)$ for each $\bi\in [q]^p_\rightarrow$. $\cC$ is a \emph{realistic} album if 
\begin{align}
\label{eqn_condition_balanced_sums}
\Pi^{\bn_\bi}_\br\ast C_\bi=\Pi_\bs^{\bn_\bj}\ast C_\bj && \mbox{ for any }&&\bi,\bj\in [q]^p_\rightarrow,\;\br,\bs\in [p]^{p-1}_\rightarrow\mbox{ such that }\bi_\br=\bj_\bs.
\end{align}
$\cC$ is a \emph{realisable} album if there exists a tensor $C\in\cT^{\bn}(\Z)$ such that $\Pi^\bn_\bi\ast C=C_\bi$ for each $\bi\in [q]^p_\rightarrow$.
\end{defn}

\begin{rem}
Crucially, the pictures in Definition~\ref{defn_album_of_pictures} are \emph{oriented}; this is enforced by taking $\bi\in [q]^p_{\rightarrow}$ instead of $\bi\in [q]^p$. Similarly, in Definition~\ref{defn_crystals}, we only require that ``oriented pictures'' of a crystal $C$ should look identical.
If we strengthened this requirement by asking that $\Pi^{n\cdot\bone_q}_\bi\ast C = M$ for \emph{all} $\bi\in [q]^2$, an $M$-crystal could only exist for
a \emph{symmetric} matrix $M$. Indeed, applying this strengthened requirement to the tuples $\bi=(i_1,i_2)$ and $\bi_{(2,1)}=(i_2,i_1)$, we would find
\begin{align*}
M
&=
\Pi^{n\cdot\bone_q}_\bi\ast C
=
\Pi^{n\cdot\bone_q}_{\bi_{(2,1)}}\ast C
\lemeq{lem_proj_contr}
\Pi^{n\cdot\bone_2}_{(2,1)}\cont{2}\Pi^{n\cdot\bone_q}_{\bi}\ast C
=
\Pi^{n\cdot\bone_2}_{(2,1)}\ast(\Pi^{n\cdot\bone_q}_{\bi}\ast C)
=
\Pi^{n\cdot\bone_2}_{(2,1)}\ast M
=M^T,
\end{align*}
where the last equality follows from the discussion in Remark~\ref{rem_1204_0307}.
This is not a sacrifice we are willing to make, as the crystal we shall need in Proposition~\ref{prop_AIPk_acceptance_graph_colouring} to fool $\AIP$ corresponds to an integer matrix having zero diagonal and whose entries sum up to one -- which, as a consequence, cannot be symmetric, see~\eqref{eqn_2018_09062022}.
\end{rem}

It is not difficult to show that, if the pictures in an album are indeed photographs of some unique tensor, then they must be compatible. In other words, a realisable album must be realistic (cf.~the beginning of the proof of Proposition~\ref{prop_realistic_equals_truthful}). Proving that a realistic album is always realisable shall require some more work.
We start by showing that the problem of checking if a realistic album is realisable does not change if we rotate the space where the tensors live. 

\begin{lem}
\label{lem_rotating_preserves_truth}
Let $p,q\in\N$, let $\bell\in [q]^q$ be such that $\#(\bell)=q$, and let
  $\bn\in\N^q$. If every realistic $(p,\bn_\bell)$-album of pictures is realisable then every realistic $(p,\bn)$-album of pictures is realisable.
\end{lem}

Proposition~\ref{prop_realistic_equals_truthful} is proved through a nested
induction, first on the dimension of the pictures (i.e., $p$), and second on the
sum of the sizes of the modes of the tensor $C$ that the pictures claim to depict (i.e.,
$\bn^T\bone_q$). Lemmas~\ref{lem_1_dim_albums_are_truthful} and~\ref{lem_realistic_realisable_case_p_1} contain the base cases for the second and the first inductions, respectively.

\begin{lem}
\label{lem_1_dim_albums_are_truthful}
A realistic $(p,\bone_q)$-album of pictures is realisable for any $p,q\in\N$.
\end{lem}

\begin{lem}
\label{lem_realistic_realisable_case_p_1}
A realistic $(1,\bn)$-album of pictures is realisable for any $q\in\N$ and $\bn\in\N^q$.
\end{lem}

\begin{prop}
\label{prop_realistic_equals_truthful}
Let $p,q\in\N$ and $\bn\in\N^q$. A $(p,\bn)$-album of pictures is realistic if and only if it is realisable.
\end{prop}

\begin{proof}[Proof of Theorem~\ref{thm_crystals_exist}]
Consider the $(2,n\cdot\bone_q)$-album of pictures
  $\cC=\{C_\bi\}_{\bi\in [q]^2_\rightarrow}$ given by $C_\bi=M$ for each
  $\bi\in [q]^2_\rightarrow$. To check that $\cC$ is a realistic album, we only
  need to notice that $\Pi^{n\cdot\bone_2}_1\ast M=M\bone_n$ and $\Pi^{n\cdot\bone_2}_2\ast
  M=M^T\bone_n$ and use that, by hypothesis, $M\bone_n=M^T\bone_n$. It then
  follows from Proposition~\ref{prop_realistic_equals_truthful} that $\cC$ is a
  realisable album. Hence, there exists a tensor
  $C\in\cT^{n\cdot\bone_q}(\Z)$ such that $\Pi^{n\cdot\bone_q}_\bi\ast C=M$ for
  each $\bi\in [q]^2_\rightarrow$. By Definition~\ref{defn_crystals}, $C$ is a $q$-dimensional $M$-crystal.
\end{proof}
The results in this section are proved in Section~\ref{appendix_proofs_on_crystals}. 
We point out that the proofs of Proposition~\ref{prop_realistic_equals_truthful} and of the lemmas needed to establish it are constructive, in that they allow to explicitly build a tensor whose projections are prescribed by a realistic album of pictures. As a consequence, the proof of Theorem~\ref{thm_crystals_exist} on the existence of crystals is constructive, too. 
We now give an example to illustrate the proof strategy.

\begin{example}
\label{example_genesis_of_a_crystal}
Throughout this example (and Example~\ref{example_crystals_fool_3_col}), we shall indicate the numbers ${\color{minusTwoColour}-2}$, ${\color{minusOneColour}-1}$, ${0}$, ${\color{plusOneColour}1}$, ${\color{plusTwoColour}2}$, and ${\color{plusThreeColour}3}$ by the colours blue, green, light grey, yellow, orange, and red, respectively. The goal is to build a $4$-dimensional 
$M$-crystal, where 
\begin{align*}
M=\begin{bmatrix}
0&0&1\\1&0&-1\\0&0&0
\end{bmatrix}=\squareTensorA{1}{.25}\,.
\end{align*}
To this end, we consider the $(2,3\cdot\bone_4)$ album of pictures $\cC$ such that all pictures are equal to 
$\squareTensorA{1}{.25}\,$. It is easy to check that $\cC$ is realistic (cf.~the proof of Theorem~\ref{thm_crystals_exist}); the goal is to show that $\cC$ is realisable, as the tensor $C\in\cT^{3\cdot\bone_4}(\Z)$ witnessing this fact would be the crystal we are looking for. 
Following the proof of Proposition~\ref{prop_realistic_equals_truthful}, we create two auxiliary albums $\hat{\cC}$ and $\tilde{\cC}$ from $\cC$. $\hat\cC$ is a $(1,3\cdot\bone_3)$-album -- i.e., both the pictures and the tensor that $\hat{\cC}$ claims to depict have one fewer dimension than those for the original album $\cC$. In particular, we see from the proof that all pictures in $\hat\cC$ are the same vector $\lineTensorA{1}{.25}\,$.
Again, it is not hard to check that $\hat\cC$ is a realistic album. To check that it is realisable, we only need to find a $3$-dimensional tensor such that summing its entries along all three modes yields $\lineTensorA{1}{.25}\,$.
Either by inspection or using the proof of Lemma~\ref{lem_realistic_realisable_case_p_1}, we find that $\hat{C}=\mbox{\cubeTensorA{-3}{.3}}\in\cT^{3\cdot\bone_3}(\Z)$ satisfies these conditions. The second album we build is the $(2,(3,3,3,2))$-album $\tilde\cC$ defined as follows: $\tilde{C}_{1,4}=\tilde{C}_{2,4}=\tilde{C}_{3,4}=\rectangularTensorA{1}{.25}\,$ (i.e., the matrix obtained by slicing off the rightmost column of $\squareTensorA{1}{.25}\,$); each other picture in the album is obtained by taking the corresponding picture in $\cC$ and subtracting from it the corresponding projection of $\hat C$ (i.e., $\tilde{C}_\bi=C_\bi-\Pi^{(3\cdot\bone_3)}_\bi\ast\hat C$). 
{\makeatletter
\let\par\@@par
\par\parshape0
\everypar{}\begin{wrapfigure}{r}{0.34\textwidth}
\crystalTensor{1}{.8}
\caption{${}$\\A $4$-dimensional {\protect\squareTensorA{1}{.25}}\hspace{-.05cm}-crystal.}
\label{fig_crystal}
\end{wrapfigure}
\noindent In this way, we obtain $\tilde C_{(1,2)}=\tilde C_{(1,3)}=\tilde C_{(2,3)}=\squareTensorB{1}{.25}\,$. This album is also realistic, and it is such that the sum of the dimensions is strictly smaller than the sum of the dimensions for the album $\cC$. At this point, we simply iterate the process, by repeatedly ``slicing'' $\tilde{\cC}$ into an album of $1$-dimensional pictures 
(which we handle through Lemma~\ref{lem_realistic_realisable_case_p_1}) and a smaller album of $2$-dimensional pictures, 
until we end up with an album such that all dimensions are shrunk to one, so that the tensor it depicts is a single number (see Lemma~\ref{lem_1_dim_albums_are_truthful}).
Throughout this process, Lemma~\ref{lem_rotating_preserves_truth} guarantees that the tensors can be rotated in a way that we slice along the rightmost mode, thus avoiding complications with the orientations of the pictures. 
In this way, we find that the album $\tilde{\cC}$ depicts the tensor $\tilde C$ whose two blocks are $\cubeTensorB{-3}{.3}\,$ and the all-zero $3\times 3\times 3$ tensor, respectively. Finally, to obtain a tensor depicted by the initial album $\cC$ (i.e., a $4$-dimensional $\squareTensorA{1}{.25}\,$-crystal), we glue together $\tilde{C}$ and $\hat{C}$. The result is shown in Figure~\ref{fig_crystal}.
\par}%
\end{example}

\subsection{Approximate graph colouring}

In this section, we prove the following result.

\begin{thm*}[Theorem~\ref{thm_AIPk_no_solves_AGC} restated]
No level of the $\AIP$ hierarchy solves  approximate graph colouring; i.e., for any fixed $3\leq c\leq d$, there is no $k$ such that the $k$-th level of $\AIP$ solves $\PCSP(\K_c,\K_d)$.  
\end{thm*}

The next proposition shows that the crystals we mined in Section~\ref{sec:quest} are able to fool the affine integer programming hierarchy. After establishing this result, Theorem~\ref{thm_AIPk_no_solves_AGC} will easily follow.

\begin{prop}
\label{prop_AIPk_acceptance_graph_colouring}
Let $k,n\in\N$ with $k\geq 2$, $n\geq 3$, and let $\GG$ be a loopless digraph. Then $\AIP^k(\GG,\K_n)=\YES$.
\end{prop}
\begin{example}
\label{example_crystals_fool_3_col}
We first illustrate Proposition~\ref{prop_AIPk_acceptance_graph_colouring} and its proof for the case $k=n=3$ and $\GG=\K_4$. 
Take the $4$-dimensional $\squareTensorA{1}{.25}\,$-crystal $C$ in Figure~\ref{fig_crystal}, and consider the map $\xi:[4]^3\to\cT^{3\cdot\bone_3}(\Z)$ defined by $\bg\mapsto\Pi^{3\cdot\bone_4}_\bg\ast C$; i.e., $\xi$ applied to a triplet $\bg$ of modes is the projection of the $4$-dimensional crystal onto the $3$-dimensional hyperplane corresponding to $\bg$. In particular, $\xi(\bg)$ is a $3\times 3\times 3$ cube. According to Theorem~\ref{prop_acceptance_AIPk}, to show that $\AIP^3(\K_4,\K_3)=\YES$, we need to prove that $\xi$ is a homomorphism from $\K_4^\tensor{3}$ to $\freeZ(\K_3^\tensor{3})$; i.e., that $\xi$ maps hyperedges of $\K_4^\tensor{3}$ to hyperedges of $\freeZ(\K_3^\tensor{3})$. (The extra condition $\xi(\bg_\bi)=\Pi^{3\cdot\bone_3}_\bi\ast \xi(\bg)$ easily follows from the definition of $\xi$). Take, for example, the hyperedge $(1,2)^\tensor{3}\in\mathcal{E}(\K_4^\tensor{3})$. Applying $\xi$ entrywise to the $2^3=8$ entries of $(1,2)^\tensor{3}$ yields the tensor $T\in\cT^{2\cdot\bone_3}(\cT^{3\cdot\bone_3}(\Z))=\cT^{6\cdot\bone_3}(\Z)$ in Figure~\ref{fig_tensor_for_3_colouring}. According to Example~\ref{ex_free_struc_tensorised_digraph}, to conclude that $T\in\mathcal{E}(\freeZ(\K_3^\tensor{3}))$, we need to exhibit some $Q\in\Zaff^{(|\mathcal{E}(\K_3)|)}=\Zaff^{(6)}$ (i.e., some integer distribution over the edges of $\K_3$) such that, for any $\bi\in [2]^3$, the $\bi$-th block of $T$ is $Q_{/\pi_\bi}$. Here it is where we use that the two-dimensional pictures of a crystal are all identical:
The $\bi$-th block of $T$ is $\xi((1,2)_\bi)=\Pi^{3\cdot\bone_4}_{(1,2)_\bi}\ast C\lemeq{lem_proj_contr}\Pi^{3\cdot\bone_2}_\bi\ast(\Pi^{3\cdot\bone_4}_{(1,2)}\ast C)=\Pi^{3\cdot\bone_2}_\bi\ast\squareTensorA{1}{.25}\,$. As a consequence, we can let $Q$ be the distribution encoded by the picture $\squareTensorA{1}{.25}\,$; i.e., the distribution assigning weight $1$ to the edges $(1,3)$ and $(2,1)$, and weight $-1$ to the edge $(2,3)$. 
\end{example}

\begin{figure}
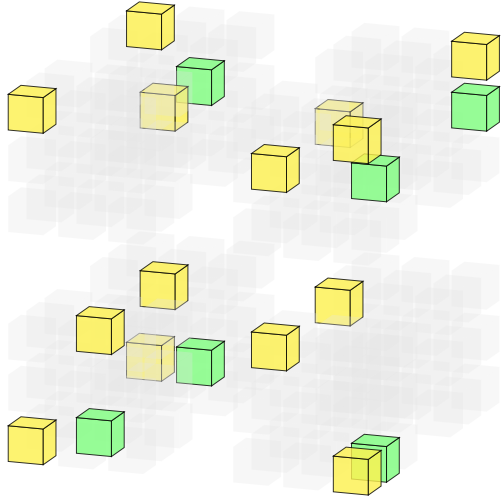

    \centering
    \finalFigure{1}{.7}
    \caption{The tensor $\xi((1,2)^\tensor{3})$. Each of the $8$ blocks is obtained by projecting the $4$-dimensional crystal from Figure~\ref{fig_crystal} onto a $3$-dimensional hyperplane.}
\label{fig_tensor_for_3_colouring}
\end{figure}

\begin{proof}[Proof of Proposition~\ref{prop_AIPk_acceptance_graph_colouring}]
Suppose, without loss of generality, that ${\mathcal{V}}(\GG)=[{q}]$ for some ${q}\in\N$. If ${q}=1$, the proposition is trivially true, so we can assume ${q}\geq 2$.
Consider the matrices
\begin{align}
\label{eqn_2018_09062022}
\hat{M}=\begin{bmatrix}
0&0&1\\1&0&-1\\0&0&0
\end{bmatrix}\in\cT^{3\cdot\bone_2}(\Z)
&&\mbox{and}&&
M=\begin{bmatrix}
\hat{M}& O\\ O&O
\end{bmatrix}\in\cT^{n\cdot\bone_2}(\Z),
\end{align}
where $O$ denotes the all-zero matrix of suitable size. Notice that
  $M\bone_n=M^T\bone_n=E_1$. (Recall that $E_i$ is the $i$-th standard unit vector of length $n$ for any $i\in [n]$.) Then, Theorem~\ref{thm_crystals_exist} provides
  us with a ${q}$-dimensional $M$-crystal $C\in\cT^{n\cdot\bone_{q}}(\Z)$.
Consider the map
\begin{align}
\label{eqn_1322_0407}
\begin{array}{rcl}
\xi:{\mathcal{V}}(\GG)^k&\to&\cT^{n\cdot\bone_k}(\Z)\\[5pt]
\bg&\mapsto& \Pi^{n\cdot\bone_q}_\bg\ast C,
\end{array}
\end{align} 
which is well defined since $\Pi^{n\cdot\bone_q}_{\bg}\in\cT^{(n\cdot\bone_k,n\cdot\bone_q)}(\Z)=\cT^{n\cdot\bone_{k+q}}(\Z)$ for any $\bg\in {\mathcal{V}}(\GG)^k$.
We claim that $\xi$ yields a homomorphism from $\GG^\tensor{k}$ to $\freeZ(\K_n^\tensor{k})$. 
First, observe that, for any $\bg\in {\mathcal{V}}(\GG)^k$,
\begin{align*}
\sum_{\ba\in [n]^k}E_\ba\ast\xi(\bg)
&\equationeq{eqn_1322_0407}
\sum_{\ba\in [n]^k}E_\ba\ast\Pi^{n\cdot\bone_q}_\bg\ast C
\lemeq{Pi_epsilon_all_one}
\Pi^{n\cdot\bone_k}_\bepsilon\ast\Pi^{n\cdot\bone_q}_\bg\ast C
\lemeq{lem_proj_contr}
\Pi^{n\cdot\bone_q}_{\bepsilon}\ast C\\
&\lemeq{lem_proj_contr}
\Pi^{n\cdot\bone_2}_{\bepsilon}\ast\Pi^{n\cdot\bone_q}_{\ang{2}}\ast C
\defeq{defn_crystals}
\Pi^{n\cdot\bone_2}_{\bepsilon}\ast M
\lemeq{Pi_epsilon_all_one}
\sum_{\bb\in [n]^2}E_\bb\ast M
=
\bone_n^TM\bone_n
=
1.
\end{align*}
Hence, $\xi(\bg)\in\Zaff^{(n^k)}={\mathcal{V}}(\freeZ(\K_n^\tensor{k}))$, as required.

We now show that $\xi$ preserves the hyperedges of $\Gk$. Recall from Definition~\ref{defn_tensorisation} that $\mathcal{E}(\Gk)=\{\bg^\tensor{k}:\bg\in\mathcal{E}(\GG)\}$.
Take $\bg\in \mathcal{E}(\GG)$; we need to prove that $\xi(\bg^\tensor{k})\in\mathcal{E}(\freeZ(\K_n\tensor{k}))$.
Observe first that $\xi(\bg^\tensor{k})=(\xi(\bg_\bi))_{\bi\in [2]^k}\in\cT^{2\cdot\bone_k}(\cT^{n\cdot\bone_k}(\Z))$.
Let $\balpha\in [2]^2$ be such that $\bg_\balpha\in [{q}]^2_\rightarrow$ (which is possible since $\#(\bg)=2$ as $\GG$ is loopless). Notice that $\balpha_\balpha=\ang{2}$.
Consider the vector ${Q}\in \cT^{n^2-n}(\Z)$ whose entries are indexed by the edges of $\K_n$ and are defined as follows: For each $\ba\in\mathcal{E}(\K_n)$, the $\ba$-th entry of ${Q}$ is $E_\ba\ast\Pi^{n\cdot\bone_2}_\balpha\ast M$. Observe that, for any $a\in [n]$, we have
\begin{align}
\label{eqn_1957_09062022}
E_{(a,a)}\ast\Pi^{n\cdot\bone_2}_\balpha\ast M
\lemeq{lem_Pi_basic}
\sum_{\substack{\bb\in [n]^2\\\bb_\balpha=(a,a)}}E_\bb\ast M
=
\sum_{\substack{\bb\in [n]^2\\\bb=(a,a)_\balpha}}E_\bb\ast M
=
E_{(a,a)}\ast M
=
0,
\end{align}
where we have used that $\balpha$ is an involution and the diagonal entries of $M$ are zero.
We find
\begin{align*}
\sum_{\ba\in \mathcal{E}({\K_n})}E_\ba\ast {{Q}} 
&=
\sum_{\ba\in \mathcal{E}({\K_n})}E_\ba\ast\Pi^{n\cdot\bone_2}_\balpha\ast M
\equationeq{eqn_1957_09062022}
\sum_{\ba\in [n]^2}E_\ba\ast\Pi^{n\cdot\bone_2}_\balpha\ast M
\lemeq{Pi_epsilon_all_one}
\Pi^{n\cdot\bone_2}_\bepsilon\ast\Pi^{n\cdot\bone_2}_\balpha\ast M\\
&\lemeq{lem_proj_contr}
\Pi^{n\cdot\bone_2}_\bepsilon\ast M
\lemeq{Pi_epsilon_all_one}
\sum_{\bb\in [n]^2}E_\bb\ast M
=
\bone_n^TM\bone_n
=
1,
\end{align*}
which means that ${{Q}}\in\Zaff^{(|\mathcal{E}({\K_n})|)}$. 
We now aim to show that $\xi(\bg_\bi)={{Q}}_{/\pi_\bi}$ for any $\bi\in [2]^k$. 
We obtain
\begin{align*}
\xi(\bg_\bi)
&\equationeq{eqn_1322_0407}
\Pi^{n\cdot\bone_q}_{\bg_\bi}\ast C
=
\Pi^{n\cdot\bone_q}_{\bg_{\balpha_{\balpha_{\bi}}}}\ast C
\lemeq{lem_proj_contr}
\Pi^{n\cdot\bone_2}_\bi\cont{2}\Pi^{n\cdot\bone_q}_{\bg_{\balpha_{\balpha}}}\ast C
\lemeq{lem_proj_contr}
\Pi^{n\cdot\bone_2}_\bi\cont{2}(\Pi^{n\cdot\bone_2}_\balpha\cont{2}\Pi^{n\cdot\bone_q}_{\bg_\balpha})\ast C\\
&=
\Pi^{n\cdot\bone_2}_\bi\ast(\Pi^{n\cdot\bone_2}_\balpha\ast(\Pi^{n\cdot\bone_q}_{\bg_\balpha}\ast C))
\defeq{defn_crystals}
\Pi^{n\cdot\bone_2}_\bi\ast(\Pi^{n\cdot\bone_2}_\balpha\ast M)
=
\Pi^{n\cdot\bone_2}_\bi\cont{2}\Pi^{n\cdot\bone_2}_\balpha\ast M.
\end{align*}
Hence, for any $\ba\in [n]^k$,
\begin{align*}
E_\ba\ast\xi(\bg_\bi)
&=
E_\ba\ast\Pi^{n\cdot\bone_2}_\bi\ast\Pi^{n\cdot\bone_2}_\balpha\ast M
\lemeq{lem_Pi_basic}
\sum_{\substack{\bb\in [n]^2\\\bb_\bi=\ba}}E_\bb\ast\Pi^{n\cdot\bone_2}_\balpha\ast M
\equationeq{eqn_1957_09062022}
\sum_{\substack{\bb\in \mathcal{E}(\K_n)\\\bb_\bi=\ba}}E_\bb\ast\Pi^{n\cdot\bone_2}_\balpha\ast M\\
&=
\sum_{\substack{\bb\in \mathcal{E}(\K_n)\\\bb_\bi=\ba}}
E_\bb\ast Q
\equationeq{eqn_1312_0407}
E_\ba\ast Q_{/\pi_\bi}.
\end{align*}
It follows that $\xi(\bg_\bi)=Q_{/\pi_\bi}$ for any $\bi\in [2]^k$, as wanted, so $\xi(\bg^\tensor{k})\in\mathcal{E}(\freeZ(\K_n\tensor{k}))$, which means that $\xi$ is indeed a homomorphism.

To be able to apply Theorem~\ref{prop_acceptance_AIPk} and conclude that $\AIP^k(\GG,\K_n)=\YES$, we are only left to observe that, for any  $\bg\in {\mathcal{V}}(\GG)^k$ and any $\bi\in [k]^k$,
\begin{align*}
\xi(\bg_\bi)
&\defeq{eqn_1322_0407}
\Pi^{n\cdot\bone_q}_{\bg_\bi}\ast C
\lemeq{lem_proj_contr}
\Pi^{n\cdot\bone_k}_\bi\cont{k}\Pi^{n\cdot\bone_q}_\bg\ast C
=
\Pi^{n\cdot\bone_k}_\bi\ast(\Pi^{n\cdot\bone_q}_\bg\ast C)
\equationeq{eqn_1322_0407}
\Pi^{n\cdot\bone_k}_\bi\ast\xi(\bg),
\end{align*}
as desired.
\end{proof}
\noindent
We remark that Proposition~\ref{prop_AIPk_acceptance_graph_colouring} does
\emph{not} hold for $n=2$, cf.~the discussion in
Appendix~\ref{sec:K2}.
\begin{proof}[Proof of Theorem~\ref{thm_AIPk_no_solves_AGC}]
Consider three integers $c,d,k$ such that $3\leq c\leq d$ and $2\leq k$.
Taking $\K_{d+1}$ as $\GG$ in
  Proposition~\ref{prop_AIPk_acceptance_graph_colouring}, 
 we find that
  $\AIP^k(\K_{d+1},\K_c)=\YES$; however, clearly, $\K_{d+1}\not\to\K_d$. Hence,
  the $k$-th level of $\AIP$ does not solve $\PCSP(\K_c,\K_d)$.
\end{proof}

\section{Proofs of results from Section~\ref{subsec_AIP_and_tensors}}
\label{sec:proofs}

\begin{lem*}[Lemma~\ref{Pi_epsilon_all_one} restated]
$E_\bepsilon=1$. Moreover,
given $q\in\N_0$ and $\bn\in\N^q$, $\Pi^\bn_\bepsilon$ is the all-one tensor in $\cT^\bn(\Z)$.
\end{lem*}
\begin{proof}
Since $\N^0=\{\bepsilon\}$ and $[\bepsilon]=\{\bepsilon\}$, $E_\bepsilon$ is well defined and lives in $\cT^\bepsilon(\Z)=\Z$. Its unique entry -- i.e., its $\bepsilon$-th entry -- is $1$ by definition.

Setting $p=0$ yields $[q]^p=\{\bepsilon\}$, so $\Pi^\bn_\bepsilon$ is well defined and lives in $\cT^{(\bn_\bepsilon,\bn)}(\Z)=\cT^\bn(\Z)$. Using that $E_\bepsilon=1$, as shown above, and applying the definition~\eqref{eqn_projection_tensor}, we find that, for any $\bb\in [\bn]$,
\begin{align*}
\Pi^\bn_\bepsilon\ast E_\bb
&=
E_\bepsilon\ast\Pi^\bn_\bepsilon\ast E_\bb
=1,
\end{align*}
as required.
\end{proof}

\begin{lem*}[Lemma~\ref{lem_Pi_basic} restated]
Given $p,q\in \N_0$, $\bn\in\N^q$, $\bi\in [q]^p$, and $\ba\in [\bn_\bi]$, we have
\begin{align*}
E_\ba\ast \Pi^\bn_\bi=\sum_{\substack{\bb\in [\bn]\\\bb_\bi=\ba}}E_{\bb}.
\end{align*}
\end{lem*}
\begin{proof}
If $p=0$, we have $\bi=\ba=\bepsilon$. Using both claims in Lemma~\ref{Pi_epsilon_all_one}, we get
\begin{align*}
E_\bepsilon\ast \Pi^\bn_\bepsilon
&=
\Pi^\bn_\bepsilon
=
\sum_{\bb\in [\bn]}E_\bb
=
\sum_{\substack{\bb\in [\bn]\\\bb_\bepsilon=\bepsilon}}E_{\bb},
\end{align*}
as required.
Suppose now that $p\in\N$. In this case, we can assume that $q\in\N$ as $[0]^p=\emptyset^p=\emptyset$.
For any $\ba'\in [\bn]$, we have
\begin{align*}
\sum_{\substack{\bb\in [\bn]\\\bb_\bi=\ba}}E_{\bb}\ast E_{\ba'}
&=
\sum_{\substack{\bb\in [\bn]\\\bb_\bi=\ba\\\bb=\ba'}}1
=
\left\{
\begin{array}{lll}
1&\mbox{ if }\ba'_\bi=\ba\\
0&\mbox{ otherwise}
\end{array}
\right.
=
E_\ba\ast\Pi^\bn_\bi\ast E_{\ba'},
\end{align*}
thus proving the result.
\end{proof}
\begin{lem*}[Lemma~\ref{lem_proj_contr} restated]
Let $m,p,q\in\N_0$, and consider two tuples $\bi\in [q]^p$ and $\bj\in [p]^m$. Then, for any $\bn\in\N^q$,
\begin{align*}
\Pi^\bn_{\bi_\bj}=\Pi^{\bn_\bi}_\bj \cont{p} \Pi^\bn_\bi.
\end{align*}
\end{lem*}
\begin{proof}
Take $\ba\in [\bn_{\bi_\bj}]$ and $\ba'\in [\bn]$, and observe that
\begin{align*}
E_\ba\ast (\Pi^{\bn_\bi}_\bj \cont{p} \Pi^\bn_\bi)\ast E_{\ba'}
&=
E_\ba\ast \Pi^{\bn_\bi}_\bj \ast \Pi^\bn_\bi\ast E_{\ba'}
\lemeq{lem_Pi_basic}
\sum_{\substack{\bb\in [\bn_\bi]\\\bb_\bj=\ba}}E_\bb\ast\Pi^\bn_\bi\ast E_{\ba'}
=
\sum_{\substack{\bb\in [\bn_\bi]\\\bb_\bj=\ba\\\ba'_\bi=\bb}}1
=
\left\{
\begin{array}{cl}
1&\mbox{ if }\ba'_{\bi_\bj}=\ba\\
0&\mbox{ otherwise}
\end{array}
\right.\\
&=
E_\ba\ast\Pi^\bn_{\bi_\bj}\ast E_{\ba'},
\end{align*}
whence the result follows.
\end{proof}

\begin{lem*}[Lemma~\ref{lem_identity_Pi} restated]
Let $q,q'\in\N_0$, $\bn\in\N^q$, $\bn'\in\N^{q'}$, and $T\in\cT^{(\bn,\bn')}(\Z)$. Then
\begin{align*}
\Pi^\bn_{\ang{q}}\cont{q} T=T.
\end{align*}
\end{lem*}
\begin{proof}
For any $\ba\in [\bn]$, we find
\begin{align*}
E_\ba\ast (\Pi^\bn_{\ang{q}}\cont{q} T)
&=
E_\ba\ast \Pi^\bn_{\ang{q}}\ast T
\lemeq{lem_Pi_basic}
\sum_{\substack{\bb\in [\bn]\\\bb_{\ang{q}}=\ba}}E_\bb\ast T
=
\sum_{\substack{\bb\in [\bn]\\\bb=\ba}}E_\bb\ast T
=
E_\ba\ast T,
\end{align*}
as required.
\end{proof}

\section{Proofs of results from Section~\ref{sec:quest}}
\label{appendix_proofs_on_crystals}

\begin{lem*}[Lemma~\ref{lem_rotating_preserves_truth} restated]
Let $p,q\in\N$, let $\bell\in [q]^q$ be such that $\#(\bell)=q$, and let
  $\bn\in\N^q$. If every realistic $(p,\bn_\bell)$-album of pictures is realisable then every realistic $(p,\bn)$-album of pictures is realisable.
\end{lem*}
\begin{proof}
Since every permutation can be expressed as the composition of inversions, it is
  enough to consider the case that $\bell$ is an inversion; in particular, $\bell_\bell=\ang{q}$.

Let $\cC=\{C_\bi\}_{\bi\in [q]^p_\rightarrow}$ be a realistic $(p,\bn)$-album of pictures. For any $\bi\in [q]^p_\rightarrow$, let $\bi^+$ be the (unique) tuple in $[p]^p$ such that $\bell_{\bi_{\bi^+}}\in [q]^p_\rightarrow$. Let also $\bi^-$ be the (unique) tuple in $[p]^p$ such that $\bi^+_{\bi^-}=\bi^-_{\bi^+}=\ang{p}$.
For each $\bi\in [q]^p_\rightarrow$, define the tensor
\begin{align}
\label{eqn_1045_0607}
\tilde{C}_\bi= \Pi^{\bn_{\bell_{\bi_{\bi^+}}}}_{\bi^-}\ast C_{\bell_{\bi_{\bi^+}}}.
\end{align}
Observe that $\tilde{C}_\bi\in\cT^{\bn_{\bell_\bi}}(\Z)$, so $\tilde{\cC}=\{\tilde{C}_\bi\}_{\bi\in [q]^p_\rightarrow}$ is a $(p,\bn_\bell)$-album of pictures. 
We claim that $\tilde{\cC}$ is a realistic album.
To prove the claim, take $\bi,\bj\in [q]^p_\rightarrow$ and $\br,\bs\in [p]^{p-1}_\rightarrow$ such that $\bi_\br=\bj_\bs$. We need to show that 
\begin{align}
\label{eqn_2041_08062022}
\Pi^{\bn_{\bell_\bi}}_{\br}\ast\tilde{C}_\bi
=
\Pi^{\bn_{\bell_\bj}}_{\bs}\ast\tilde{C}_\bj.
\end{align}
Let $\balpha,\bbeta\in [p-1]^{p-1}$ be the (unique) tuples such that $\bi^-_{\br_\balpha}\in [p]^{p-1}_\rightarrow$ and $\balpha_\bbeta=\bbeta_\balpha=\ang{p-1}$. We claim that $\bj^-_{\bs_\balpha}\in [p]^{p-1}_\rightarrow$. Indeed, for any $x,y\in [p-1]$ such that $x<y$ we have 
\begin{align*}
i^-_{r_{\alpha_x}}<i^-_{r_{\alpha_y}}
\quad
&\Rightarrow
\quad
\ell_{i_{i^+_{i^-_{r_{\alpha_x}}}}}
<
\ell_{i_{i^+_{i^-_{r_{\alpha_y}}}}}
\quad
\Rightarrow
\quad
\ell_{i_{{{r_{\alpha_x}}}}}
<
\ell_{i_{{{r_{\alpha_y}}}}}
\quad
\Rightarrow
\quad
\ell_{j_{{{s_{\alpha_x}}}}}
<
\ell_{j_{{{s_{\alpha_y}}}}}\\
\quad
&\Rightarrow
\quad
\ell_{j_{j^+_{j^-_{s_{\alpha_x}}}}}
<
\ell_{j_{j^+_{j^-_{s_{\alpha_y}}}}}
\quad
\Rightarrow
\quad
j^-_{s_{\alpha_x}}<j^-_{s_{\alpha_y}},
\end{align*}
thus proving the claim.
Therefore,
\begin{align}
\label{eqn_2117_0806_2022_A}
\notag
\Pi^{\bn_{\bell_\bi}}_{\br}\ast\tilde{C}_\bi
&\equationeq{eqn_1045_0607}
\Pi^{\bn_{\bell_\bi}}_{\br}\ast\left(\Pi^{\bn_{\bell_{\bi_{\bi^+}}}}_{\bi^-}\ast C_{\bell_{\bi_{\bi^+}}}\right)
=
\Pi^{\bn_{\bell_\bi}}_{\br}\cont{p}\Pi^{\bn_{\bell_{\bi_{\bi^+}}}}_{\bi^-}\ast C_{\bell_{\bi_{\bi^+}}}
\lemeq{lem_proj_contr}
\Pi^{\bn_{\bell_{\bi_{\bi^+}}}}_{\bi^-_\br}\ast C_{\bell_{\bi_{\bi^+}}}
=
\Pi^{\bn_{\bell_{\bi_{\bi^+}}}}_{\bi^-_{\br_{\balpha_\bbeta}}}\ast C_{\bell_{\bi_{\bi^+}}}\\
&\lemeq{lem_proj_contr}
\Pi^{\bn_{\bell_{\bi_{\br_\balpha}}}}_\bbeta\cont{p-1}\Pi^{\bn_{\bell_{\bi_{\bi^+}}}}_{\bi^-_{\br_{\balpha}}}\ast C_{\bell_{\bi_{\bi^+}}}
=
\Pi^{\bn_{\bell_{\bi_{\br_\balpha}}}}_\bbeta\ast\left(\Pi^{\bn_{\bell_{\bi_{\bi^+}}}}_{\bi^-_{\br_{\balpha}}}\ast C_{\bell_{\bi_{\bi^+}}}\right)
\intertext{and, similarly,}
\label{eqn_2117_0806_2022_B}
\Pi^{\bn_{\bell_\bj}}_{\bs}\ast\tilde{C}_\bj
&=
\Pi^{\bn_{\bell_{\bj_{\bs_\balpha}}}}_\bbeta\ast\left(\Pi^{\bn_{\bell_{\bj_{\bj^+}}}}_{\bj^-_{\bs_{\balpha}}}\ast C_{\bell_{\bj_{\bj^+}}}\right).
\end{align}
Let us now focus on the tuples $\bell_{\bi_{\bi^+}},\bell_{\bj_{\bj+}}\in [q]^p_\rightarrow$ and $\bi^-_{\br_\balpha},\bj^-_{\bs_\balpha}\in [p]^{p-1}_\rightarrow$. Observe that 
\begin{align*}
\bell_{\bi_{\bi^+_{\bi^-_{\br_\balpha}}}}
&=
\bell_{\bi_{{{\br_\balpha}}}}
=
\bell_{\bj_{{{\bs_\balpha}}}}
=
\bell_{\bj_{\bj^+_{\bj^-_{\bs_\balpha}}}}.
\end{align*}
Using that $\cC$ is a realistic album, we deduce
\begin{align}
\label{eqn_2117_0806_2022_C}
\Pi^{\bn_{\bell_{\bi_{\bi^+}}}}_{\bi^-_{\br_\balpha}}\ast C_{\bell_{\bi_{\bi^+}}}
&=
\Pi^{\bn_{\bell_{\bj_{\bj+}}}}_{\bj^-_{\bs_\balpha}}\ast C_{\bell_{\bj_{\bj+}}}.
\end{align}
Combining~\eqref{eqn_2117_0806_2022_A},~\eqref{eqn_2117_0806_2022_B}, and~\eqref{eqn_2117_0806_2022_C}, and recalling that $\bi_\br=\bj_\bs$, yields~\eqref{eqn_2041_08062022}, thus proving that $\tilde{\cC}$ is a realistic $(p,\bn_\bell)$-album of pictures, as claimed. From the hypothesis of the lemma, we deduce that $\tilde{\cC}$ is realisable, so there exists a tensor $\tilde{C}\in\cT^{\bn_\bell}(\Z)$ such that $\Pi^{\bn_\bell}_\bi\ast\tilde{C}=\tilde{C}_\bi$ for each $\bi\in [q]^p_\rightarrow$. Define $C= \Pi^{\bn_\bell}_{\bell}\ast\tilde{C}\in\cT^{\bn}(\Z)$ (where we are using that $\bell_\bell=\ang{q}$). Given $\bi\in [q]^p_\rightarrow$, we find
\begin{align}
\label{eqn_2158_08062022_A}
\begin{array}{lclclcl}
\displaystyle \Pi^\bn_\bi\ast C
&=&
\displaystyle \Pi^\bn_\bi\ast (\Pi^{\bn_\bell}_{\bell}\ast\tilde{C})
&=&
\displaystyle \Pi^\bn_{\bi_{\bi^+_{\bi^-}}}\ast (\Pi^{\bn_\bell}_{\bell}\ast\tilde{C})
&\lemeq{lem_proj_contr}&
\displaystyle \Pi^{\bn_{\bi_{\bi^+}}}_{\bi^-}\cont{p}\Pi^{\bn_{\bi}}_{\bi^+}\cont{p}\Pi^\bn_\bi\cont{q} \Pi^{\bn_\bell}_{\bell}\ast\tilde{C}\\[15pt]
&
\lemeq{lem_proj_contr}&
\displaystyle \Pi^{\bn_{\bi_{\bi^+}}}_{\bi^-}\cont{p}
\Pi^{\bn_\bell}_{\bell_{\bi_{\bi^+}}}\ast\tilde{C}
&=&
\displaystyle \Pi^{\bn_{\bi_{\bi^+}}}_{\bi^-}\ast(
\Pi^{\bn_\bell}_{\bell_{\bi_{\bi^+}}}\ast\tilde{C})
&=&
\displaystyle \Pi^{\bn_{\bi_{\bi^+}}}_{\bi^-}\ast\tilde{C}_{\bell_{\bi_{\bi^+}}}.
\end{array}
\end{align}
Notice that $\bell_{\bell_{\bi_{\bi^+_{\bi^-}}}}=\bi$, which is an increasing tuple. Hence, $\left(\bell_{\bi_{\bi^+}}\right)^+=\bi^-$ and, consequently, $\left(\bell_{\bi_{\bi^+}}\right)^-=\bi^+$. It follows from~\eqref{eqn_1045_0607} that
\begin{align}
\label{eqn_2158_08062022_B}
\tilde{C}_{\bell_{\bi_{\bi^+}}}
&=
\Pi^{\bn_{\bell_{\bell_{\bi_{\bi^+_{\bi^-}}}}}}_{\bi^+}\ast C_{\bell_{\bell_{\bi_{\bi^+_{\bi^-}}}}}
=
\Pi^{\bn_{\bi}}_{\bi^+}\ast C_{\bi}.
\end{align}
Combining~\eqref{eqn_2158_08062022_A} and~\eqref{eqn_2158_08062022_B} yields
\begin{align*}
\Pi^\bn_\bi\ast C
&=
\Pi^{\bn_{\bi_{\bi^+}}}_{\bi^-}\ast(\Pi^{\bn_{\bi}}_{\bi^+}\ast C_{\bi})
=
\Pi^{\bn_{\bi_{\bi^+}}}_{\bi^-}\cont{p}\Pi^{\bn_{\bi}}_{\bi^+}\ast C_{\bi}
\lemeq{lem_proj_contr}
\Pi^{\bn_\bi}_{\ang{p}}\ast C_\bi
\lemeq{lem_identity_Pi}
C_\bi,
\end{align*}
which concludes the proof that $\cC$ is a realisable album of pictures.
\end{proof}

\begin{lem*}[Lemma~\ref{lem_1_dim_albums_are_truthful} restated]
A realistic $(p,\bone_q)$-album of pictures is realisable for any $p,q\in\N$.
\end{lem*}
\begin{proof}
Let $\cC=\{C_\bi\}_{\bi\in [q]^p_\rightarrow}$ be a realistic $(p,\bone_q)$-album of pictures. For any $\bi\in [q]^p_\rightarrow$, $C_\bi\in\cT^{(\bone_q)_\bi}(\Z)=\cT^{\bone_p}(\Z)$. We claim that $C_\bi=C_\bj$ for any $\bi,\bj\in [q]^p_\rightarrow$.  Define, for each pair $\bi,\bj\in [q]^p_\rightarrow$, their \emph{distance} $\operatorname{d}(\bi,\bj)$ as the cardinality of the set $\{t\in [p]:i_t\neq j_t\}$. Suppose, for the sake of contradiction, that the claim is false, and let  $\bi,\bj\in [q]^p_\rightarrow$ attain the minimum distance among all pairs $\bi',\bj'$ for which $C_{\bi'}\neq C_{\bj'}$. Let $\alpha=\max\{t\in [p]:i_t\neq j_t\}$. Assume, without loss of generality, that $i_\alpha<j_\alpha$, and define a new tuple $\bell\in [q]^p$ obtained from $\bi$ by replacing $i_\alpha$ with $j_\alpha$. Observe that $i_1<i_2<\dots<i_{\alpha-1}<i_\alpha<j_\alpha<j_{\alpha+1}=i_{\alpha+1}<i_{\alpha+2}<\dots<i_p$, so $\bell\in [q]^p_\rightarrow
$. Letting $\br\in [p]^{p-1}_\rightarrow$ be obtained from $\ang{p}$ by deleting its $\alpha$-th entry, observe that $\bi_\br=\bell_\br$. Using that $\cC$ is a realistic album, we obtain $\Pi^{\bone_p}_\br\ast C_\bi=\Pi_\br^{\bone_p}\ast C_\bell$. Therefore,
\begin{align*}
E_{\bone_p}\ast C_\bi
&\lemeq{lem_Pi_basic}
E_{\bone_{p-1}}\ast \Pi^{\bone_p}_\br\ast C_\bi
=
E_{\bone_{p-1}}\ast \Pi^{\bone_p}_\br\ast C_\bell
\lemeq{lem_Pi_basic}
E_{\bone_p}\ast C_\bell,
\end{align*} 
so $C_\bell=C_\bi\neq C_\bj$. But this contradicts the choice of the pair $(\bi,\bj)$, as $\operatorname{d}(\bell,\bj)=\operatorname{d}(\bi,\bj)-1$. Hence, the claim is true.
We can then define a tensor $C\in\cT^{\bone_q}(\Z)$ by setting $E_{\bone_q}\ast C=E_{\bone_p}\ast C_\bi$ for any $\bi\in [q]^p_\rightarrow$. In this way, we get
\begin{align*}
E_{\bone_p}\ast\Pi^{\bone_q}_\bi\ast C
\lemeq{lem_Pi_basic}
E_{\bone_q}\ast C
=
E_{\bone_p}\ast C_\bi.
\end{align*}
We conclude that $\Pi^{\bone_q}_\bi\ast C=C_\bi$ for any $\bi\in [q]^p_\rightarrow$, which means that $\cC$ is a realisable album.
\end{proof}

\begin{lem*}[Lemma~\ref{lem_realistic_realisable_case_p_1} restated]
A realistic $(1,\bn)$-album of pictures is realisable for any $q\in\N$ and $\bn\in\N^q$.
\end{lem*}
\begin{proof}
If $q=1$, the result is trivially true, so we assume $q\geq 2$.
Notice that $ [q]^1_\rightarrow=[q]$, so each element of $[q]^1_\rightarrow$ is a single number.
We prove the statement by induction on $\bn^T\bone_q$.
If $\bn^T\bone_q=q$, then $\bn=\bone_q$, and the result follows from Lemma~\ref{lem_1_dim_albums_are_truthful}. Suppose that $\bn^T\bone_q\geq q+1$. 
Using Lemma~\ref{lem_rotating_preserves_truth}, we can assume $n_q\geq 2$ without losing generality. Let $\cC=\{C_i\}_{i\in [q]}$ be a realistic $(1,\bn)$-album of pictures; observe that $C_i$ is a vector in $\cT^{n_i}(\Z)$ for each $i\in [q]$.
Set $\ell=E_{n_q}\ast C_q$ (i.e., $\ell$ is the last entry of $C_q$), and consider a new family of tensors $\tilde\cC=\{\tilde C_i\}_{i\in [q]}$ defined by 
\begin{align*}
\tilde{C}_i=
\left\{
\begin{array}{lll}
C_i-\ell E_{n_i} & \mbox{if }i\in [q-1]\\
(E_1\ast C_q,\dots,E_{n_q-1}\ast C_q)  & \mbox{if }i=q.
\end{array}
\right.
\end{align*}
Let $\tilde{\bn}=\bn-E_q$ and notice that $\tilde{\bn}\in\N^q$ since $n_q\geq 2$. We have that $C_i\in\cT^{\tilde{n}_i}(\Z)$ for each $i\in [q]$, so $\tilde\cC$ is a $(1,\tilde{\bn})$-album of pictures. 

We now show that $\tilde\cC$ is realistic.
By definition, $[1]^{0}_\rightarrow=\{\bepsilon\}$, so we only need to show that
\begin{align}
\label{eqn_1739_06062022_A}
\Pi^{\tilde n_i}_\bepsilon\ast \tilde C_i=
\Pi^{\tilde n_j}_\bepsilon\ast \tilde C_j
&& \forall\, i,j\in [q].
\end{align} 
We claim that 
\begin{align}
\label{eqn_1739_06062022_B}
\Pi^{\tilde n_i}_\bepsilon\ast \tilde C_i=\Pi^{n_i}_\bepsilon\ast C_i-\ell && &\forall\, i\in [q].
\end{align}
Then,~\eqref{eqn_1739_06062022_A} would follow from the fact that $\cC$ is a realistic album. If $i\in [q-1]$,
\begin{align*}
\Pi^{\tilde n_i}_\bepsilon\ast \tilde C_i
&=
\Pi^{n_i}_\bepsilon\ast (C_i-\ell E_{n_i})
\lemeq{Pi_epsilon_all_one}
\Pi^{n_i}_\bepsilon\ast C_i-\ell,
\intertext{so~\eqref{eqn_1739_06062022_B} holds in this case. Moreover,}
\Pi^{\tilde n_q}_\bepsilon\ast \tilde C_q
&=
\Pi^{n_q-1}_\bepsilon\ast (E_1\ast C_q,\dots,E_{n_q-1}\ast C_q)
\lemeq{Pi_epsilon_all_one}
\sum_{b\in [n_q-1]}E_b\ast C_q 
=
\bone_{n_q}\ast C_q-\ell
\lemeq{Pi_epsilon_all_one}
\Pi^{n_q}_\bepsilon\ast C_q-\ell,
\end{align*}
so~\eqref{eqn_1739_06062022_B} holds in this case as well. We conclude that $\tilde\cC$ is indeed a realistic album. 

Since $\tilde\bn^T\bone_q=\bn^T\bone_q-1$, we have from the inductive hypothesis that $\tilde{\cC}$ is realisable, so there exists a tensor $\tilde C\in\cT^{\tilde{\bn}}(\Z)$ such that $\Pi^{\tilde{\bn}}_i\ast \tilde{C}=\tilde{C}_i$ for each $i\in [q]$. Define a tensor $C\in\cT^{\bn}(\Z)$ by setting, for each $\bb\in [\bn]$,
\begin{align}
\label{eqn_1230_0507}
E_\bb\ast C=
\left\{
\begin{array}{llll}
\ell&\mbox{ if }\bb=\bn\\
0&\mbox{ if }\bb\neq \bn\;\mbox{ and }\;b_q=n_q\\
E_\bb\ast \tilde{C}&\mbox{ if }b_q\neq n_q.
\end{array}
\right.
\end{align}
(Notice that the last line of the right-hand side of the above expression is well defined as, if $b_q\neq n_q$, then $\bb\in [\tilde{\bn}]$.)
Take $i\in [q]$; we claim that $\Pi^\bn_i\ast C=C_i$. For $a\in [n_i]$, we find
\begin{align*}
E_a\ast \Pi^\bn_i\ast C
&\lemeq{lem_Pi_basic}
\sum_{\substack{\bb\in [\bn]\\b_i=a}}E_\bb\ast C.
\end{align*}
For $i\neq q$, this yields
\begin{align*}
E_a\ast \Pi^\bn_i\ast C
&=
\sum_{\substack{\bb\in [\bn]\\b_i=a\\b_q=n_q}}E_\bb\ast C
+
\sum_{\substack{\bb\in [\bn]\\b_i=a\\b_q\neq n_q}}E_\bb\ast C
\equationeq{eqn_1230_0507}
\ell\cdot \delta_{a,n_i}
+
\sum_{\substack{\bb\in [\tilde{\bn}]\\b_i=a}}E_\bb\ast \tilde C
\intertext{(where $\delta_{a,n_i}$ is $1$ if $a=n_i$, $0$ otherwise)}
&\lemeq{lem_Pi_basic}
\ell\cdot \delta_{a,n_i}
+
E_a\ast \Pi^{\tilde{\bn}}_i\ast \tilde{C}
=
\ell\cdot \delta_{a,n_i}
+
E_a\ast\tilde{C}_i
=
\ell\cdot \delta_{a,n_i}
+
E_a\ast(C_i-\ell E_{n_i})
=
E_a\ast C_i.
\end{align*}
For $i=q$, if $a=n_q$ we get
\begin{align*}
E_a\ast \Pi^\bn_q\ast C
&=
\sum_{\substack{\bb\in [\bn]\\b_q=n_q}}E_\bb\ast C
\equationeq{eqn_1230_0507}
\ell
=
E_a\ast C_q,
\end{align*}
while if $a\neq n_q$ we get 
\begin{align*}
E_a\ast \Pi^\bn_q\ast C
&=
\sum_{\substack{\bb\in [\bn]\\b_q=a}}E_\bb\ast C
\equationeq{eqn_1230_0507}
\sum_{\substack{\bb\in [\tilde{\bn}]\\b_q=a}}E_\bb\ast \tilde C
\lemeq{lem_Pi_basic}
E_a\ast\Pi^{\tilde{\bn}}_q\ast\tilde{C}
=
E_a\ast\tilde{C}_q\\
&=
E_a\ast (E_1\ast C_q,\dots,E_{n_q-1}\ast C_q)
=
E_a\ast C_q.
\end{align*} 
It follows that $\Pi^\bn_i\ast C=C_i$, as claimed. Therefore, $\cC$ is a realisable album.
\end{proof}

\begin{prop*}[Proposition~\ref{prop_realistic_equals_truthful} restated]
Let $p,q\in\N$ and $\bn\in\N^q$. A $(p,\bn)$-album of pictures is realistic if and only if it is realisable.
\end{prop*}
\begin{proof}
Let $\cC=\{C_\bi\}_{\bi\in [q]^p_\rightarrow}$ be a realisable album of pictures; i.e., there exists $C\in\cT^\bn(\Z)$ such that $\Pi^\bn_\bi\ast C= C_\bi$ for each $\bi\in [q]^p_\rightarrow$. For any $\bi,\bj\in [q]^p_\rightarrow$ and $\br,\bs\in [p]^{p-1}_\rightarrow$ such that $\bi_\br=\bj_\bs$, we find
\begin{align*}
\Pi^{\bn_\bi}_\br\ast C_\bi
&=
\Pi^{\bn_\bi}_\br\ast(\Pi^\bn_\bi\ast C)
=
\Pi^{\bn_\bi}_\br\cont{p}\Pi^\bn_\bi\ast C
\lemeq{lem_proj_contr}
\Pi^{\bn}_{\bi_\br}\ast C
=
\Pi^{\bn}_{\bj_\bs}\ast C
\lemeq{lem_proj_contr}
\Pi^{\bn_\bj}_\bs\cont{p}\Pi^\bn_\bj\ast C
=
\Pi^{\bn_\bj}_\bs\ast C_\bj,
\end{align*}
which shows that $\cC$ is a realistic album. Hence, the ``if'' part of the statement is true. Next, we focus on the ``only if'' part.

We prove the result by nested induction, first on $p$ and second on $\bn^T\bone_q$. For $p=1$, the result follows from Lemma~\ref{lem_realistic_realisable_case_p_1}. Suppose that $p\geq 2$. 
For $\bn^T\bone_q=q$ (which implies $\bn=\bone_q$), the result follows from Lemma~\ref{lem_1_dim_albums_are_truthful}. Suppose that $\bn^T\bone_q\geq q+1$. Using Lemma~\ref{lem_rotating_preserves_truth}, we can safely assume $n_q\geq 2$. 
If $q=1$, then $[q]^p_\to=\emptyset$ and the statement is trivially true, so we can assume $q\geq 2$.
Let $\cC=\{C_\bi\}_{\bi\in [q]^p_\to}$ be a realistic $(p,\bn)$-album of pictures; we need to show that $\cC$ is realisable.

Set
$\hat{\bn}=(n_1,\dots,n_{q-1})\in\N^{q-1}$. For any $\bi\in [q-1]^{p-1}_\rightarrow$, we define $\hat C_\bi\in \cT^{\hat \bn_\bi}(\Z)$ by $E_\ba\ast \hat{C}_\bi= E_{(\ba,n_q)}\ast C_{(\bi,q)}$ for each $\ba\in [\hat{\bn}_\bi]$. Observe that the last expression is well defined, as $\bi\in [q-1]^{p-1}_\rightarrow$ implies that $(\bi,q)\in [q]^{p}_\rightarrow$. 
We claim that the family $\hat\cC=\{\hat{C}_\bi\}_{\bi\in [q-1]^{p-1}_\rightarrow}$ is a realistic $(p-1,\hat\bn)$-album of pictures. Take $\bi,\bj\in [q-1]^{p-1}_\rightarrow$ and $\br,\bs\in [p-1]^{p-2}_\rightarrow$ such that $\bi_\br=\bj_\bs$. For any $\ba\in [\hat{\bn}_{\bi_\br}]$, we find
\begin{align}
\label{eqn_1656_0507_A}
\notag
E_\ba\ast\Pi^{\hat{\bn}_\bi}_\br\ast\hat{C}_\bi
&\lemeq{lem_Pi_basic}
\sum_{\substack{\bb\in [\hat{\bn}_\bi]\\ \bb_\br=\ba}}E_\bb\ast\hat{C}_\bi
=
\sum_{\substack{\bb\in [\hat{\bn}_\bi]\\ \bb_\br=\ba}}
E_{(\bb,n_q)}\ast C_{(\bi,q)}
=
\sum_{\substack{\bc\in [{\bn}_{(\bi,q)}]\\ \bc_{(\br,p)}=(\ba,n_q)}}
E_{\bc}\ast C_{(\bi,q)}\\
&\lemeq{lem_Pi_basic}
E_{(\ba,n_q)}\ast\Pi^{\bn_{(\bi,q)}}_{(\br,p)}\ast C_{(\bi,q)}
\intertext{and, similarly,}
\label{eqn_1656_0507_B}
E_\ba\ast\Pi^{\hat{\bn}_\bj}_\bs\ast\hat{C}_\bj
&=
E_{(\ba,n_q)}\ast\Pi^{\bn_{(\bj,q)}}_{(\bs,p)}\ast C_{(\bj,q)}.
\end{align}
We now use the fact that $\cC$ is a realistic album. In particular, we apply~\eqref{eqn_condition_balanced_sums} to the tuples $(\bi,q),(\bj,q)\in [q]^p_\rightarrow$ and $(\br,p),(\bs,p)\in [p]^{p-1}_\rightarrow$ (note that $(\bi,q)_{(\br,p)}=(\bi_\br,q)=(\bj_\bs,q)=(\bj,q)_{(\bs,p)}$). Since $(\ba,n_q)\in [\bn_{(\bi,q)_{(\br,p)}}]$, we obtain
\begin{align*}
E_{(\ba,n_q)}\ast\Pi^{\bn_{(\bi,q)}}_{(\br,p)}\ast C_{(\bi,q)}
&=
E_{(\ba,n_q)}\ast\Pi^{\bn_{(\bj,q)}}_{(\bs,p)}\ast C_{(\bj,q)}.
\end{align*}
Combining this with~\eqref{eqn_1656_0507_A} and~\eqref{eqn_1656_0507_B} yields
\begin{align*}
E_\ba\ast\Pi^{\hat{\bn}_\bi}_\br\ast\hat{C}_\bi
&=
E_\ba\ast\Pi^{\hat{\bn}_\bj}_\bs\ast\hat{C}_\bj.
\end{align*}
We conclude that $\hat\cC$ is a realistic album, as claimed.
It follows from the inductive hypothesis that $\hat\cC$ is realisable, so we can find a tensor $\hat{C}\in\cT^{\hat\bn}(\Z)$ such that $\Pi^{\hat{\bn}}_\bi\ast\hat{C}=\hat{C}_\bi$ for each $\bi\in [q-1]^{p-1}_\rightarrow$. 
Let now $\tilde{\bn}=\bn-E_q\in\N^q$. For any $\bi\in [q]^p_\rightarrow$, define a tensor $\tilde{C}_\bi\in\cT^{\tilde{\bn}_\bi}(\Z)$ as follows:
 If $i_p\neq q$ (in which case $\bi\in [q-1]^p_\rightarrow$) we set $\tilde{C}_\bi=C_\bi-\Pi^{\hat{\bn}}_\bi\ast\hat C$; if $i_p=q$, for $\bb\in [\tilde{\bn}_\bi]$, we set $E_\bb\ast\tilde{C}_\bi=E_\bb\ast C_\bi$ (where the last expression is well defined as $[\tilde{\bn}]\subseteq [\bn]$, so $[\tilde{\bn}_\bi]\subseteq [\bn_\bi]$).
We claim that the family $\tilde\cC=\{\tilde C_\bi\}_{\bi\in [q]^p_\rightarrow}$ is a realistic $(p,\tilde{\bn})$-album of pictures. To that end, we shall first prove that the equation
\begin{align}
\label{eqn_1417_08062022}
E_\ba\ast
\Pi^{\tilde\bn_\bi}_\br\ast \tilde C_\bi
=
\left\{
\begin{array}{llll}
E_\ba\ast\Pi^{\bn_\bi}_\br\ast C_\bi&\mbox{ if }i_{r_{p-1}}= q\\
E_\ba\ast(\Pi^{\bn_\bi}_\br\ast C_\bi-\hat{C}_{\bi_\br})&\mbox{ otherwise}
\end{array}
\right.
\end{align} 
is satisfied for any $\bi\in [q]^p_\rightarrow$, any $\br\in [p]^{p-1}_\rightarrow$, and any $\ba\in [\tilde{\bn}_{\bi_\br}]$. First, notice that, if $i_p=q$,
\begin{align*}
[\tilde{\bn}_\bi]
&=
[\tilde{n}_{i_1}]\times\dots\times [\tilde{n}_{i_{p-1}}]\times [\tilde{n}_{i_p}]
=
[{n}_{i_1}]\times\dots\times [n_{i_{p-1}}]\times [{n}_{q}-1]
=
\{\bb\in [\bn_\bi]:b_p\neq n_q\}
\end{align*}
while, if $i_p\neq q$, $\tilde{\bn}_\bi=\hat\bn_\bi=\bn_\bi$, so $[\tilde{\bn}_\bi]=[\hat\bn_\bi]=[\bn_\bi]$.
Suppose that $i_{r_{p-1}}=q$. In this case, we have $r_{p-1}=p$ and $i_p=q$. Hence, 
\begin{align*}
E_\ba\ast\Pi^{\tilde\bn_\bi}_\br\ast \tilde C_\bi
&\lemeq{lem_Pi_basic} 
\sum_{\substack{\bb\in [\tilde{\bn}_\bi]\\ \bb_\br=\ba}}
E_\bb\ast\tilde{C}_\bi
=
\sum_{\substack{\bb\in [{\bn}_\bi]\\ \bb_\br=\ba\\ b_p\neq n_q}}
E_\bb\ast C_\bi
=
\sum_{\substack{\bb\in [{\bn}_\bi]\\ \bb_\br=\ba}}
E_\bb\ast C_\bi
\lemeq{lem_Pi_basic}
E_\ba\ast\Pi^{\bn_\bi}_{\br}\ast C_\bi,
\end{align*}
so~\eqref{eqn_1417_08062022} holds in this case. Suppose now that $i_{r_{p-1}}\neq q$. This can happen either if $i_p\neq q$ (\underline{case \emph{a}}), or if $i_p=q$ and $r_{p-1}\neq p$ (\underline{case \emph{b}}), and it implies that $\bi_\br\in [q-1]^{p-1}_\rightarrow$. In \underline{case \emph{a}}, 
\begin{align*}
\Pi^{\tilde\bn_\bi}_\br\ast \tilde C_\bi
&=
\Pi^{\bn_\bi}_\br\ast (C_\bi-\Pi^{\hat\bn}_\bi\ast\hat{C})
=
\Pi^{\bn_\bi}_\br\ast C_\bi - \Pi^{\bn_\bi}_\br\cont{p}\Pi^{\hat\bn}_\bi\ast\hat{C}
\lemeq{lem_proj_contr}
\Pi^{\bn_\bi}_\br\ast C_\bi - \Pi^{\hat\bn}_{\bi_\br}\ast\hat{C}
=
\Pi^{\bn_\bi}_\br\ast C_\bi - \hat{C}_{\bi_\br},
\end{align*}
where the last equality follows from the property of $\hat{C}$. So,~\eqref{eqn_1417_08062022} holds in this case. In \underline{case \emph{b}}, we must have $\br=\ang{p-1}$. Hence,
\begin{align*}
E_\ba\ast\Pi^{\tilde\bn_\bi}_\br\ast \tilde C_\bi
&\lemeq{lem_Pi_basic}
\sum_{\substack{\bb\in [\tilde{\bn}_\bi]\\ \bb_{\ang{p-1}}=\ba}}
E_\bb\ast\tilde{C}_\bi
=
\sum_{\substack{\bb\in [{\bn}_\bi]\\ \bb_{\ang{p-1}}=\ba\\ b_p\neq n_q}}
E_\bb\ast C_\bi
=
\sum_{\substack{\bb\in [{\bn}_\bi]\\ \bb_{\ang{p-1}}=\ba}}
E_\bb\ast C_\bi
-
E_{(\ba,n_q)}\ast C_\bi\\
&\lemeq{lem_Pi_basic}
E_\ba\ast\Pi^{\bn_\bi}_{\ang{p-1}}\ast C_\bi-E_{(\ba,n_q)}\ast C_\bi
=
E_\ba\ast\Pi^{\bn_\bi}_{\ang{p-1}}\ast C_\bi-E_{(\ba,n_q)}\ast C_{(\bi_{\ang{p-1}},q)}\\
&=
E_\ba\ast\Pi^{\bn_\bi}_{\ang{p-1}}\ast C_\bi-E_\ba\ast \hat C_{\bi_{\ang{p-1}}}
=
E_\ba\ast(\Pi^{\bn_\bi}_{\br}\ast C_\bi-\hat C_{\bi_{\br}})
,
\end{align*}
where the penultimate equality comes from the definition of $\hat\cC$ and from the fact that, in this case, $\tilde{\bn}_{\bi_\br}=\hat\bn_{\bi_\br}$, so $\ba\in [\hat{\bn}_{\bi_\br}]$. We conclude that~\eqref{eqn_1417_08062022} also holds in \underline{case \emph{b}}. Using~\eqref{eqn_1417_08062022} and the fact that $\cC$ is a realistic album, we easily conclude that $\tilde{\cC}$ is a realistic album, too. Indeed, take $\bi,\bj\in [q]^p_\rightarrow$ and $\br,\bs\in [p]^{p-1}_\rightarrow$ such that $\bi_\br=\bj_\bs$, and choose $\ba\in [\tilde{\bn}_{\bi_\br}]$. Observe that $i_{r_{p-1}}=j_{s_{p-1}}$. If $i_{r_{p-1}}=q$, we find
\begin{align*}
E_\ba\ast\Pi^{\tilde{\bn}_\bi}_\br\ast\tilde{C}_\bi
&=
E_\ba\ast\Pi^{\bn_\bi}_\br\ast C_\bi
=
E_\ba\ast\Pi^{\bn_\bj}_\bs\ast C_\bj
=
E_\ba\ast\Pi^{\tilde{\bn}_\bj}_\bs\ast\tilde{C}_\bj;
\intertext{otherwise,}
E_\ba\ast\Pi^{\tilde{\bn}_\bi}_\br\ast\tilde{C}_\bi
&=
E_\ba\ast(\Pi^{\bn_\bi}_\br\ast C_\bi-\hat{C}_{\bi_\br})
=
E_\ba\ast(\Pi^{\bn_\bj}_\bs\ast C_\bj-\hat{C}_{\bj_\bs})
=
E_\ba\ast\Pi^{\tilde{\bn}_\bj}_\bs\ast\tilde{C}_\bj.
\end{align*}
It follows that $\tilde{\cC}$ is indeed a realistic album, as claimed.
Since $\tilde{\bn}^T\bone_q=\bn^T\bone_q-1$, we can then apply the inductive hypothesis to deduce that $\tilde\cC$ is realisable, so there exists a tensor $\tilde{C}\in\cT^{\tilde{\bn}}(\Z)$ such that $\Pi^{\tilde{\bn}}_\bi\ast\tilde{C}=\tilde{C}_\bi$ for each $\bi\in [q]^p_\rightarrow$. 

We now define a tensor $C\in\cT^{\bn}(\Z)$ by setting, for each $\bb\in [\bn]$,
\begin{align}
\label{eqn_2124_0507}
E_\bb\ast C=
\left\{
\begin{array}{llll}
E_{\bb_{\langle q-1\rangle}}\ast \hat{C}&\mbox{ if }& b_q=n_q\\[5pt]
E_{\bb}\ast \tilde{C}&\mbox{ if }&b_q\neq n_q.
\end{array}
\right.
\end{align}
Take $\bi\in [q]^p_\rightarrow$ and $\ba\in [\bn_\bi]$. To conclude the proof, we need to show that 
\begin{align}
\label{eqn_goal_1259_07062022}
E_\ba\ast\Pi^\bn_\bi\ast C=E_\ba\ast C_\bi.
\end{align}
Observe that 
\begin{align}
\label{eqn_1311_07062022}
E_\ba\ast\Pi^\bn_\bi\ast C
&\lemeq{lem_Pi_basic}
\sum_{\substack{\bb\in [\bn]\\ \bb_\bi=\ba}}E_\bb\ast C
=
\sum_{\substack{\bb\in [\bn]\\ \bb_\bi=\ba\\ b_q=n_q}}E_\bb\ast C
+
\sum_{\substack{\bb\in [\bn]\\ \bb_\bi=\ba\\ b_q\neq n_q}}E_\bb\ast C
\equationeq{eqn_2124_0507}
\sum_{\substack{\bb\in [\bn]\\ \bb_\bi=\ba\\ b_q=n_q}}E_{\bb_{\ang{q-1}}}\ast \hat C
+
\sum_{\substack{\bb\in [\tilde\bn]\\ \bb_\bi=\ba}}E_\bb\ast \tilde C.
\end{align}
Let us denote the first and the second summand of the rightmost expression in~\eqref{eqn_1311_07062022} by $\alpha$ and $\beta$, respectively.
Suppose first that $i_p=q$. If $a_p\neq n_q$, we see that $\alpha=0$, so
\begin{align*}
E_\ba\ast\Pi^\bn_\bi\ast C
&\equationeq{eqn_1311_07062022}
\sum_{\substack{\bb\in [\tilde\bn]\\ \bb_\bi=\ba}}E_\bb\ast \tilde C
\lemeq{lem_Pi_basic}
E_\ba\ast\Pi^{\tilde\bn}_\bi\ast \tilde{C}
=
E_\ba\ast\tilde{C}_\bi
=
E_\ba\ast C_\bi; 
\end{align*} 
if $a_p=n_q$, we get $\beta=0$, so  
\begin{align*}
E_\ba\ast\Pi^\bn_\bi\ast C
&\equationeq{eqn_1311_07062022}
\sum_{\substack{\bb\in [\bn]\\ \bb_\bi=\ba\\ b_q=n_q}}E_{\bb_{\ang{q-1}}}\ast \hat C
=
\sum_{\substack{\bb\in [\bn]\\ \bb_\bi=\ba}}E_{\bb_{\ang{q-1}}}\ast \hat C
=
\sum_{\substack{\bc\in [\hat \bn]\\\bc_{\bi_{\ang{p-1}}}=\ba_{\ang{p-1}}}}E_\bc\ast\hat{C}
\lemeq{lem_Pi_basic}
E_{\ba_{\ang{p-1}}}\ast\Pi^{\hat\bn}_{\bi_{\ang{p-1}}}\ast \hat{C}\\
&=
E_{\ba_{\ang{p-1}}}\ast\hat{C}_{\bi_{\ang{p-1}}}
=
E_{(\ba_{\ang{p-1}},n_q)}\ast C_{(\bi_{\ang{p-1}},q)}
=
E_\ba\ast C_\bi.
\end{align*}
Suppose now that $i_p\neq q$, in which case $\bi\in [q-1]^p_\rightarrow$. We obtain
\begin{align*}
\alpha
&=
\sum_{\substack{\bb\in [\bn]\\ \bb_\bi=\ba\\ b_q=n_q}}E_{\bb_{\ang{q-1}}}\ast \hat C
=
\sum_{\substack{\bc\in [\hat{\bn}]\\ \bc_{\bi}=\ba}}E_\bc\ast\hat{C}
\lemeq{lem_Pi_basic}
E_\ba\ast\Pi^{\hat{\bn}}_{\bi}\ast\hat C,\\
\beta
&=
\sum_{\substack{\bb\in [\tilde\bn]\\ \bb_\bi=\ba}}E_\bb\ast \tilde C
\lemeq{lem_Pi_basic}
E_\ba\ast\Pi^{\tilde{\bn}}_{\bi}\ast\tilde{C} 
=
E_\ba\ast\tilde{C}_\bi
=
E_\ba\ast(C_\bi-\Pi^{\hat\bn}_\bi\ast\hat{C})
=
E_\ba\ast C_\bi - E_\ba\ast\Pi^{\hat\bn}_\bi\ast\hat{C},
\end{align*}
and it follows that
\begin{align*}
E_\ba\ast\Pi^\bn_\bi\ast C
&\equationeq{eqn_1311_07062022}
\alpha+\beta
=
E_\ba\ast C_\bi.
\end{align*}
Therefore,~\eqref{eqn_goal_1259_07062022} holds, $\cC$ is realisable, and the proof is concluded.\qedhere
\end{proof}

\appendix
\section{Affine integer programming and $2$-colouring}
\label{sec:K2}

Recall the statement of Proposition~\ref{prop_AIPk_acceptance_graph_colouring}.

\begin{prop*}[Proposition~\ref{prop_AIPk_acceptance_graph_colouring} restated]
Let $k,n\in\N$ with $k\geq 2$, $n\geq 3$, and let $\GG$ be a loopless digraph. Then $\AIP^k(\GG,\K_n)=\YES$.
\end{prop*}

Proposition~\ref{prop_AIPk_acceptance_graph_colouring} does not hold for $n=2$; i.e., for the $2$-clique $\K_2$. The reason why the proof does not work in this case is that the argument relies on the existence of an $n\times n$ ``picture-matrix'' $M$ that $(i)$ is entrywise integer and affine (i.e., its entries sum up to one), $(ii)$ has zero diagonal, and $(iii)$ satisfies $M\bone_n=M^T\bone_n$. The purpose of $(i)$ is to make $M$ the $2$-dimensional projection of a proper solution to the $k$-th level of $\AIP$, the purpose of $(ii)$ is to make the solution compatible with the edges of $\K_n$, and the purpose of $(iii)$ is to make sure that the album of pictures consisting of copies of $M$ is realistic, so that an $M$-crystal exists. For $n\geq 3$,~\eqref{eqn_2018_09062022} provides an example of such matrices; for $n\leq 2$, however, one readily checks that there is no matrix satisfying all three requirements.

In fact, one can show using the algebraic machinery developed in~\cite{BBKO21} that the base level of the $\AIP$ hierarchy is already powerful enough to solve
$\CSP(\K_2)=\PCSP(\K_2,\K_2)$. 
Let the symbol ``$=_2$'' mean ``equal mod $2$''. 
\begin{defn}[\cite{BBKO21}]
Let $A,B$ be sets and let $L\geq 3$ be an odd integer. A function $f:A^{L}\to B$ is said to be \emph{alternating} if
\begin{itemize}
\item
$f(\ba)=f(\ba_\bi)$ for any $\ba\in A^L$ and any $\bi\in [L]^L$ such that $\#(\bi)=L$ and $\bi$ preserves the parity (i.e., $i_j=_2 j$ $\forall j\in [L]$), and
\item
$f(\ba,b,b)=f(\ba,c,c)$ for any $\ba\in A^{L-2}$ and any $b,c\in A$.
\end{itemize}
\end{defn}
Given $L\in\N$ and two digraphs $\HH,\tilde\HH$ such that $\HH\to\tilde{\HH}$, an \emph{$L$-ary polymorphism} from $\HH$ to $\tilde\HH$ is a homomorphism from $\HH^L$ to $\tilde\HH$, where $\HH^L$ is the $L$-th \emph{direct power} of $\HH$ -- i.e., the digraph whose vertex set is $\mathcal{V}(\HH^L)=\mathcal{V}(\HH)^L$ and whose edge set is $\mathcal{E}(\HH^L)=\{(\bh,\bell):\bh,\bell\in\mathcal{V}(\HH)^L \mbox{ and }(h_i,\ell_i)\in\mathcal{E}(\HH) \;\forall i\in [L]\}$.

The next result characterises the power of the affine integer programming relaxation (denoted by $\AIP$) as defined in~\cite{BBKO21}, which essentially corresponds to the case $k=1$ of the hierarchy considered in this work. (In particular, $\AIP$ is weaker than any level of the $\AIP$ hierarchy, in the sense that any $\PCSP$ template solved by $\AIP$ is also solved by any level of the hierarchy.)

\begin{thm}[\cite{BBKO21}]
\label{thm_solvability_AIP}
$\AIP$ solves $\PCSP(\HH,\tilde{\HH})$ if and only if there exist $L$-ary alternating polymorphisms from $\HH$ to $\tilde\HH$ for any odd $L\geq 3$. 
\end{thm}  
Let the vertex set of $\K_2$ be $\mathcal{V}(\K_2)=\{0,1\}$.
Given an odd integer $L\geq 3$, consider the function $f:\{0,1\}^L\to\{0,1\}$ defined by $\bh\mapsto\sum_{i=1}^L(-1)^{i+1}h_i \mod 2$, and notice that it is alternating. We claim that $f$ is a polymorphism from $\K_2$ to itself -- i.e., a homomorphism from $\K_2^L$ to $\K_2$. Take $\bh,\bell\in \{0,1\}^L$ and suppose that $(\bh,\bell)\in\mathcal{E}(\K_2^L)$; i.e., $h_i+\ell_i=1$ for each $i\in [L]$.
If $f(\bh)=f(\bell)$, we find 
\begin{align*}
f(\bh)
&=_2
\sum_{i=1}^L(-1)^{i+1}h_i
=_2
\sum_{i=1}^L(-1)^{i+1}\ell_i=_2
\sum_{i=1}^L(-1)^{i+1}(1-h_i)
=_2
\sum_{i=1}^L(-1)^{i+1}-\sum_{i=1}^L(-1)^{i+1}h_i\\
&=_2 
1-f(\bh),
\end{align*} 
so $2f(\bh)=_2 1$, a contradiction. As a consequence, $f(\bh)\neq f(\bell)$, so $f((\bh,\bell))\in\mathcal{E}(\K_2)$. Hence, $f$ is a polymorphism, as claimed.
Using Theorem~\ref{thm_solvability_AIP}, it follows that $\AIP$ solves $\CSP(\K_2)$, which means that
$\AIP(\GG,\K_2)$ 
outputs $\YES$ exactly when $\GG$ is $2$-colourable, i.e., bipartite.

{\small
\bibliographystyle{plainurl}
\bibliography{cz}

\begin{thebibliography}{10}

\bibitem{Arora06:toc}
Sanjeev Arora, B{\'{e}}la Bollob{\'{a}}s, L{\'{a}}szl{\'{o}} Lov{\'{a}}sz, and
  Iannis Tourlakis.
\newblock {Proving Integrality Gaps without Knowing the Linear Program}.
\newblock {\em Theory Comput.}, 2(2):19--51, 2006.
\newblock \href {https://doi.org/10.4086/toc.2006.v002a002}
  {\path{doi:10.4086/toc.2006.v002a002}}.

\bibitem{AB21}
Kristina Asimi and Libor Barto.
\newblock Finitely tractable promise constraint satisfaction problems.
\newblock In {\em Proc. 46th International Symposium on Mathematical
  Foundations of Computer Science (MFCS'21)}, volume 202 of {\em LIPIcs}, pages
  11:1--11:16. Schloss Dagstuhl -- Leibniz-Zentrum f{\"{u}}r Informatik, 2021.
\newblock \href {https://doi.org/10.4230/LIPIcs.MFCS.2021.11}
  {\path{doi:10.4230/LIPIcs.MFCS.2021.11}}.

\bibitem{Atserias22:soda}
Albert Atserias and V{\'{\i}}ctor Dalmau.
\newblock {Promise Constraint Satisfaction and Width}.
\newblock In {\em Proc. 2022 ACM-SIAM Symposium on Discrete Algorithms
  (SODA'22)}, pages 1129--1153, 2022.
\newblock \href {http://arxiv.org/abs/2107.05886} {\path{arXiv:2107.05886}},
  \href {https://doi.org/10.1137/1.9781611977073.48}
  {\path{doi:10.1137/1.9781611977073.48}}.

\bibitem{AGH17}
Per Austrin, Venkatesan Guruswami, and Johan H{\aa}stad.
\newblock (2+{$\epsilon$})-{S}at is {NP}-hard.
\newblock {\em {SIAM} J. Comput.}, 46(5):1554--1573, 2017.
\newblock \href {http://arxiv.org/abs/2013/159} {\path{arXiv:2013/159}}, \href
  {https://doi.org/10.1137/15M1006507} {\path{doi:10.1137/15M1006507}}.

\bibitem{Barto21:stacs}
Libor Barto, Diego Battistelli, and Kevin~M. Berg.
\newblock {Symmetric Promise Constraint Satisfaction Problems: Beyond the
  Boolean Case}.
\newblock In {\em Proc. 38th International Symposium on Theoretical Aspects of
  Computer Science (STACS'21)}, volume 187 of {\em LIPIcs}, pages 10:1--10:16.
  Schloss Dagstuhl -- Leibniz-Zentrum f{\"u}r Informatik, 2021.
\newblock \href {http://arxiv.org/abs/2010.04623} {\path{arXiv:2010.04623}},
  \href {https://doi.org/10.4230/LIPIcs.STACS.2021.10}
  {\path{doi:10.4230/LIPIcs.STACS.2021.10}}.

\bibitem{BBKO21}
Libor Barto, Jakub Bul{\'{\i}}n, Andrei~A. Krokhin, and Jakub Opr\v{s}al.
\newblock Algebraic approach to promise constraint satisfaction.
\newblock {\em J. {ACM}}, 68(4):28:1--28:66, 2021.
\newblock \href {http://arxiv.org/abs/1811.00970} {\path{arXiv:1811.00970}},
  \href {https://doi.org/10.1145/3457606} {\path{doi:10.1145/3457606}}.

\bibitem{Barto22:soda}
Libor Barto and Marcin Kozik.
\newblock {Combinatorial Gap Theorem and Reductions between Promise CSPs}.
\newblock In {\em Proc. 2022 ACM-SIAM Symposium on Discrete Algorithms
  (SODA'22)}, pages 1204--1220, 2022.
\newblock \href {http://arxiv.org/abs/2107.09423} {\path{arXiv:2107.09423}},
  \href {https://doi.org/10.1137/1.9781611977073.50}
  {\path{doi:10.1137/1.9781611977073.50}}.

\bibitem{MR2890890}
Alexander Barvinok.
\newblock Matrices with prescribed row and column sums.
\newblock {\em Linear Algebra Appl.}, 436(4):820--844, 2012.
\newblock \href {https://doi.org/10.1016/j.laa.2010.11.019}
  {\path{doi:10.1016/j.laa.2010.11.019}}.

\bibitem{Berkholz17:soda}
Christoph Berkholz and Martin Grohe.
\newblock {Linear Diophantine Equations, Group CSPs, and Graph Isomorphism}.
\newblock In {\em Proc. 28th Annual {ACM-SIAM} Symposium on Discrete Algorithms
  (SODA'17)}, pages 327--339. {SIAM}, 2017.
\newblock \href {http://arxiv.org/abs/1607.04287} {\path{arXiv:1607.04287}},
  \href {https://doi.org/10.1137/1.9781611974782.21}
  {\path{doi:10.1137/1.9781611974782.21}}.

\bibitem{Bhangale21:stoc}
Amey Bhangale and Subhash Khot.
\newblock {Optimal Inapproximability of Satisfiable k-LIN over Non-Abelian
  Groups}.
\newblock In {\em Proc. 53rd Annual ACM Symposium on Theory of Computing
  (STOC'21)}, pages 1615--1628. {ACM}, 2021.
\newblock \href {http://arxiv.org/abs/2009.02815} {\path{arXiv:2009.02815}},
  \href {https://doi.org/10.1145/3406325.3451003}
  {\path{doi:10.1145/3406325.3451003}}.

\bibitem{Bhangale22:stoc}
Amey Bhangale, Subhash Khot, and Don Minzer.
\newblock {On Inapproximability of Satisfiable $k$-CSPs: I.}
\newblock In {\em Proc. 54th Annual ACM Symposium on Theory of Computing
  (STOC'22)}, pages 976--988. {ACM}, 2022.
\newblock \href {https://doi.org/10.1145/3519935.3520028}
  {\path{doi:10.1145/3519935.3520028}}.

\bibitem{BrakensiekG16}
Joshua Brakensiek and Venkatesan Guruswami.
\newblock New hardness results for graph and hypergraph colorings.
\newblock In {\em Proc. 31st Conference on Computational Complexity (CCC'16)},
  volume~50 of {\em LIPIcs}, pages 14:1--14:27. Schloss Dagstuhl -
  Leibniz-Zentrum f{\"{u}}r Informatik, 2016.
\newblock \href {https://doi.org/10.4230/LIPIcs.CCC.2016.14}
  {\path{doi:10.4230/LIPIcs.CCC.2016.14}}.

\bibitem{BG19}
Joshua Brakensiek and Venkatesan Guruswami.
\newblock An algorithmic blend of {LP}s and ring equations for promise {CSP}s.
\newblock In {\em Proc. 30th Annual {ACM-SIAM} Symposium on Discrete Algorithms
  (SODA'19)}, pages 436--455, 2019.
\newblock \href {http://arxiv.org/abs/1807.05194} {\path{arXiv:1807.05194}},
  \href {https://doi.org/10.1137/1.9781611975482.28}
  {\path{doi:10.1137/1.9781611975482.28}}.

\bibitem{BG21:sicomp}
Joshua Brakensiek and Venkatesan Guruswami.
\newblock {Promise Constraint Satisfaction: Algebraic Structure and a Symmetric
  Boolean Dichotomy}.
\newblock {\em {SIAM} J. Comput.}, 50(6):1663--1700, 2021.
\newblock \href {http://arxiv.org/abs/1704.01937} {\path{arXiv:1704.01937}},
  \href {https://doi.org/10.1137/19M128212X} {\path{doi:10.1137/19M128212X}}.

\bibitem{BGS21}
Joshua Brakensiek, Venkatesan Guruswami, and Sai Sandeep.
\newblock {Conditional Dichotomy of Boolean Ordered Promise CSPs}.
\newblock In {\em Proc. 48th International Colloquium on Automata, Languages,
  and Programming (ICALP'21)}, volume 198 of {\em LIPIcs}, pages 37:1--37:15.
  Schloss Dagstuhl -- Leibniz-Zentrum f{\"u}r Informatik, 2021.
\newblock \href {http://arxiv.org/abs/2102.11854} {\path{arXiv:2102.11854}},
  \href {https://doi.org/10.4230/LIPIcs.ICALP.2021.37}
  {\path{doi:10.4230/LIPIcs.ICALP.2021.37}}.

\bibitem{bgwz20}
Joshua Brakensiek, Venkatesan Guruswami, Marcin Wrochna, and Stanislav
  {\v{Z}}ivn{\'y}.
\newblock The power of the combined basic {LP} and affine relaxation for
  promise {CSP}s.
\newblock {\em {SIAM} J. Comput.}, 49:1232--1248, 2020.
\newblock \href {http://arxiv.org/abs/1907.04383} {\path{arXiv:1907.04383}},
  \href {https://doi.org/10.1137/20M1312745} {\path{doi:10.1137/20M1312745}}.

\bibitem{BWZ21}
Alex Brandts, Marcin Wrochna, and Stanislav {\v{Z}}ivn{\'y}.
\newblock The complexity of promise {SAT} on non-{B}oolean domains.
\newblock {\em {ACM} Trans. Comput. Theory}, 13(4):26:1--26:20, 2021.
\newblock \href {http://arxiv.org/abs/1911.09065} {\path{arXiv:1911.09065}},
  \href {https://doi.org/10.1145/3470867} {\path{doi:10.1145/3470867}}.

\bibitem{Braun15:stoc}
G{\'{a}}bor Braun, Sebastian Pokutta, and Daniel Zink.
\newblock {Inapproximability of Combinatorial Problems via Small LPs and SDPs}.
\newblock In {\em Proc. 47th Annual {ACM} on Symposium on Theory of Computing
  (STOC'15)}, pages 107--116. {ACM}, 2015.
\newblock \href {https://doi.org/10.1145/2746539.2746550}
  {\path{doi:10.1145/2746539.2746550}}.

\bibitem{Braverman21:focs}
Mark Braverman, Subhash Khot, Noam Lifshitz, and Dor Minzer.
\newblock {An Invariance Principle for the Multi-slice, with Applications}.
\newblock In {\em Proc. 62nd {IEEE} Annual Symposium on Foundations of Computer
  Science (FOCS'21)}, pages 228--236. {IEEE}, 2021.
\newblock \href {https://doi.org/10.1109/FOCS52979.2021.00030}
  {\path{doi:10.1109/FOCS52979.2021.00030}}.

\bibitem{Braverman21:itcs}
Mark Braverman, Subhash Khot, and Dor Minzer.
\newblock On rich 2-to-1 games.
\newblock In {\em Proc. 12th Innovations in Theoretical Computer Science
  Conference (ITCS'21)}, volume 185 of {\em LIPIcs}, pages 27:1--27:20. Schloss
  Dagstuhl - Leibniz-Zentrum f{\"{u}}r Informatik, 2021.
\newblock \href {https://doi.org/10.4230/LIPIcs.ITCS.2021.27}
  {\path{doi:10.4230/LIPIcs.ITCS.2021.27}}.

\bibitem{MR3759214}
Richard~A. Brualdi and Geir Dahl.
\newblock Alternating sign matrices and hypermatrices, and a generalization of
  {L}atin squares.
\newblock {\em Adv. in Appl. Math.}, 95:116--151, 2018.
\newblock \href {https://doi.org/10.1016/j.aam.2017.11.005}
  {\path{doi:10.1016/j.aam.2017.11.005}}.

\bibitem{brualdi2021sign}
Richard~A Brualdi and Geir Dahl.
\newblock Sign-restricted matrices of $0$'s, $1$'s, and $-1$'s.
\newblock {\em Linear Algebra and its Applications}, 615:77--103, 2021.
\newblock \href {https://doi.org/10.1016/j.laa.2021.01.001}
  {\path{doi:10.1016/j.laa.2021.01.001}}.

\bibitem{brualdi1991combinatorial}
Richard~A Brualdi and Herbert~J Ryser.
\newblock {\em Combinatorial matrix theory}, volume~39.
\newblock Springer, 1991.

\bibitem{Bulatov17:focs}
Andrei~A. Bulatov.
\newblock A dichotomy theorem for nonuniform {CSP}s.
\newblock In {\em Proc. 58th Annual IEEE Symposium on Foundations of Computer
  Science (FOCS'17)}, pages 319--330, 2017.
\newblock \href {http://arxiv.org/abs/1703.03021} {\path{arXiv:1703.03021}},
  \href {https://doi.org/10.1109/FOCS.2017.37}
  {\path{doi:10.1109/FOCS.2017.37}}.

\bibitem{Butti21:mfcs}
Silvia Butti and V{\'{\i}}ctor Dalmau.
\newblock {Fractional Homomorphism, Weisfeiler-Leman Invariance, and the
  Sherali-Adams Hierarchy for the Constraint Satisfaction Problem}.
\newblock In {\em Proc. 46th International Symposium on Mathematical
  Foundations of Computer Science (MFCS'21)}, volume 202 of {\em LIPIcs}, pages
  27:1--27:19. Schloss Dagstuhl -- Leibniz-Zentrum f{\"{u}}r Informatik, 2021.
\newblock \href {http://arxiv.org/abs/2107.02956} {\path{arXiv:2107.02956}},
  \href {https://doi.org/10.4230/LIPIcs.MFCS.2021.27}
  {\path{doi:10.4230/LIPIcs.MFCS.2021.27}}.

\bibitem{Chan15:jacm}
Siu~On Chan.
\newblock {Approximation Resistance from Pairwise-Independent Subgroups}.
\newblock {\em J. {ACM}}, 63(3):27:1--27:32, 2016.
\newblock \href {https://doi.org/10.1145/2873054} {\path{doi:10.1145/2873054}}.

\bibitem{Chan16:jacm-lp}
Siu~On Chan, James~R. Lee, Prasad Raghavendra, and David Steurer.
\newblock {Approximate Constraint Satisfaction Requires Large {LP}
  Relaxations}.
\newblock {\em J. {ACM}}, 63(4):34:1--34:22, 2016.
\newblock \href {https://doi.org/10.1145/2811255} {\path{doi:10.1145/2811255}}.

\bibitem{CZ22:sicomp-clap}
Lorenzo Ciardo and Stanislav \v{Z}ivn\'y.
\newblock {CLAP: A New Algorithm for Promise CSPs}.
\newblock {\em SIAM Journal on Computing}, 2022.
\newblock \href {http://arxiv.org/abs/2107.05018} {\path{arXiv:2107.05018}}.

\bibitem{CZ22minions}
Lorenzo Ciardo and Stanislav \v{Z}ivn\'y.
\newblock {Hierarchies of minion tests for PCSPs through tensors}.
\newblock In {\em Proc. 2023 ACM-SIAM Symposium on Discrete Algorithms
  (SODA'23)}, 2023.
\newblock To appear.
\newblock \href {http://arxiv.org/abs/2207.02277} {\path{arXiv:2207.02277}}.

\bibitem{da20090}
Carlos~M da~Fonseca and Ricardo Mamede.
\newblock On $(0,1)$-matrices with prescribed row and column sum vectors.
\newblock {\em Discrete mathematics}, 309(8):2519--2527, 2009.
\newblock \href {https://doi.org/10.1016/j.disc.2008.06.013}
  {\path{doi:10.1016/j.disc.2008.06.013}}.

\bibitem{Dinur18:stoc-non-optimality}
Irit Dinur, Subhash Khot, Guy Kindler, Dor Minzer, and Muli Safra.
\newblock On non-optimally expanding sets in {G}rassmann graphs.
\newblock In {\em Proc. 50th Annual {ACM} {SIGACT} Symposium on Theory of
  Computing (STOC'18)}, pages 940--951. {ACM}, 2018.
\newblock \href {https://doi.org/10.1145/3188745.3188806}
  {\path{doi:10.1145/3188745.3188806}}.

\bibitem{Dinur18:stoc-towards}
Irit Dinur, Subhash Khot, Guy Kindler, Dor Minzer, and Muli Safra.
\newblock Towards a proof of the 2-to-1 games conjecture?
\newblock In {\em Proc. 50th Annual {ACM} {SIGACT} Symposium on Theory of
  Computing (STOC'18)}, pages 376--389. {ACM}, 2018.
\newblock \href {https://doi.org/10.1145/3188745.3188804}
  {\path{doi:10.1145/3188745.3188804}}.

\bibitem{Dinur09:sicomp}
Irit Dinur, Elchanan Mossel, and Oded Regev.
\newblock {Conditional Hardness for Approximate Coloring}.
\newblock {\em {SIAM} J. Comput.}, 39(3):843--873, 2009.
\newblock \href {https://doi.org/10.1137/07068062X}
  {\path{doi:10.1137/07068062X}}.

\bibitem{Feder98:monotone}
Tom\'as Feder and Moshe~Y. Vardi.
\newblock The {C}omputational {S}tructure of {M}onotone {M}onadic {S{N}{P}} and
  {C}onstraint {S}atisfaction: {A} {S}tudy through {D}atalog and {G}roup
  {T}heory.
\newblock {\em {SIAM} J. Comput.}, 28(1):57--104, 1998.
\newblock \href {https://doi.org/10.1137/S0097539794266766}
  {\path{doi:10.1137/S0097539794266766}}.

\bibitem{GJ76}
M.~R. Garey and David~S. Johnson.
\newblock The complexity of near-optimal graph coloring.
\newblock {\em J. {ACM}}, 23(1):43--49, 1976.
\newblock \href {https://doi.org/10.1145/321921.321926}
  {\path{doi:10.1145/321921.321926}}.

\bibitem{Ghosh18:toc}
Mrinalkanti Ghosh and Madhur Tulsiani.
\newblock {From Weak to Strong Linear Programming Gaps for All Constraint
  Satisfaction Problems}.
\newblock {\em Theory Comput.}, 14(1):1--33, 2018.
\newblock \href {https://doi.org/10.4086/toc.2018.v014a010}
  {\path{doi:10.4086/toc.2018.v014a010}}.

\bibitem{GK04}
Venkatesan Guruswami and Sanjeev Khanna.
\newblock On the hardness of 4-coloring a 3-colorable graph.
\newblock {\em SIAM Journal on Discrete Mathematics}, 18(1):30--40, 2004.
\newblock \href {https://doi.org/10.1137/S0895480100376794}
  {\path{doi:10.1137/S0895480100376794}}.

\bibitem{GS20:icalp}
Venkatesan Guruswami and Sai Sandeep.
\newblock {d-To-1 Hardness of Coloring 3-Colorable Graphs with {O(1)} Colors}.
\newblock In {\em Proc. 47th International Colloquium on Automata, Languages,
  and Programming (ICALP'20)}, volume 168 of {\em LIPIcs}, pages 62:1--62:12.
  Schloss Dagstuhl -- Leibniz-Zentrum f{\"u}r Informatik, 2020.
\newblock \href {https://doi.org/10.4230/LIPIcs.ICALP.2020.62}
  {\path{doi:10.4230/LIPIcs.ICALP.2020.62}}.

\bibitem{Huang13}
Sangxia Huang.
\newblock Improved hardness of approximating chromatic number.
\newblock In {\em Proc. 16th International Workshop on Approximation,
  Randomization, and Combinatorial Optimization. Algorithms and Techniques and
  the 17th International Workshop on Randomization and Computation
  (APPROX-RANDOM'13)}, pages 233--243. Springer, 2013.
\newblock \href {http://arxiv.org/abs/1301.5216} {\path{arXiv:1301.5216}},
  \href {https://doi.org/10.1007/978-3-642-40328-6_17}
  {\path{doi:10.1007/978-3-642-40328-6_17}}.

\bibitem{kannan1979}
Ravindran Kannan and Achim Bachem.
\newblock Polynomial {{Algorithms}} for {{Computing}} the {{Smith}} and
  {{Hermite Normal Forms}} of an {{Integer Matrix}}.
\newblock {\em SIAM Journal on Computing}, 8(4):499--507, November 1979.
\newblock \href {https://doi.org/10.1137/0208040} {\path{doi:10.1137/0208040}}.

\bibitem{Karp72}
Richard~M. Karp.
\newblock {Reducibility Among Combinatorial Problems}.
\newblock In {\em {Proc. Complexity of Computer Computations}}, pages 85--103,
  1972.
\newblock \href {https://doi.org/10.1007/978-1-4684-2001-2\_9}
  {\path{doi:10.1007/978-1-4684-2001-2\_9}}.

\bibitem{KhannaLS00}
Sanjeev Khanna, Nathan Linial, and Shmuel Safra.
\newblock On the hardness of approximating the chromatic number.
\newblock {\em Comb.}, 20(3):393--415, 2000.
\newblock \href {https://doi.org/10.1007/s004930070013}
  {\path{doi:10.1007/s004930070013}}.

\bibitem{Khot01}
Subhash Khot.
\newblock {Improved Inaproximability Results for MaxClique, Chromatic Number
  and Approximate Graph Coloring}.
\newblock In {\em Proc. 42nd Annual IEEE Symposium on Foundations of Computer
  Science (FOCS'01)}, pages 600--609. {IEEE} Computer Society, 2001.
\newblock \href {https://doi.org/10.1109/SFCS.2001.959936}
  {\path{doi:10.1109/SFCS.2001.959936}}.

\bibitem{Khot02stoc}
Subhash Khot.
\newblock On the power of unique 2-prover 1-round games.
\newblock In {\em Proc. 34th Annual ACM Symposium on Theory of Computing
  ({STOC}'02)}, pages 767--775. {ACM}, 2002.
\newblock \href {https://doi.org/10.1145/509907.510017}
  {\path{doi:10.1145/509907.510017}}.

\bibitem{Khot17:stoc-independent}
Subhash Khot, Dor Minzer, and Muli Safra.
\newblock On independent sets, 2-to-2 games, and {G}rassmann graphs.
\newblock In {\em Proc. 49th Annual {ACM} {SIGACT} Symposium on Theory of
  Computing (STOC'17)}, pages 576--589. {ACM}, 2017.
\newblock \href {https://doi.org/10.1145/3055399.3055432}
  {\path{doi:10.1145/3055399.3055432}}.

\bibitem{Khot18:focs-pseudorandom}
Subhash Khot, Dor Minzer, and Muli Safra.
\newblock Pseudorandom sets in {G}rassmann graph have near-perfect expansion.
\newblock In {\em Proc. 59th {IEEE} Annual Symposium on Foundations of Computer
  Science (FOCS'18)}, pages 592--601. {IEEE} Computer Society, 2018.
\newblock \href {https://doi.org/10.1109/FOCS.2018.00062}
  {\path{doi:10.1109/FOCS.2018.00062}}.

\bibitem{Kothari17:stoc}
Pravesh~K. Kothari, Raghu Meka, and Prasad Raghavendra.
\newblock {Approximating rectangles by juntas and weakly-exponential lower
  bounds for {LP} relaxations of CSPs}.
\newblock In {\em Proc. 49th Annual {ACM} {SIGACT} Symposium on Theory of
  Computing (STOC'17)}, pages 590--603. {ACM}, 2017.
\newblock \href {https://doi.org/10.1145/3055399.3055438}
  {\path{doi:10.1145/3055399.3055438}}.

\bibitem{Lasserre02}
Jean~B. Lasserre.
\newblock An explicit equivalent positive semidefinite program for nonlinear
  0-1 programs.
\newblock {\em {SIAM} J. Optim.}, 12(3):756--769, 2002.
\newblock \href {https://doi.org/10.1137/S1052623400380079}
  {\path{doi:10.1137/S1052623400380079}}.

\bibitem{Lee15:stoc}
James~R. Lee, Prasad Raghavendra, and David Steurer.
\newblock {Lower Bounds on the Size of Semidefinite Programming Relaxations}.
\newblock In {\em Proc. 47th Annual {ACM} on Symposium on Theory of Computing
  (STOC'15)}, pages 567--576. {ACM}, 2015.
\newblock \href {https://doi.org/10.1145/2746539.2746599}
  {\path{doi:10.1145/2746539.2746599}}.

\bibitem{NZ22:icalp}
Tamio{-}Vesa Nakajima and Stanislav \v{Z}ivn{\'{y}}.
\newblock Linearly ordered colourings of hypergraphs.
\newblock In {\em Proc. 49th International Colloquium on Automata, Languages,
  and Programming (ICALP'22)}, volume 229 of {\em LIPIcs}, pages 128:1--128:18.
  Schloss Dagstuhl - Leibniz-Zentrum f{\"{u}}r Informatik, 2022.
\newblock \href {http://arxiv.org/abs/2204.05628} {\path{arXiv:2204.05628}},
  \href {https://doi.org/10.4230/LIPIcs.ICALP.2022.128}
  {\path{doi:10.4230/LIPIcs.ICALP.2022.128}}.

\bibitem{parrilo2000structured}
Pablo~A Parrilo.
\newblock {\em Structured semidefinite programs and semialgebraic geometry
  methods in robustness and optimization}.
\newblock California Institute of Technology, 2000.
\newblock URL: \url{http://www.cds.caltech.edu/~doyle/hot/thesis.pdf}.

\bibitem{ryser1957combinatorial}
Herbert~J Ryser.
\newblock Combinatorial properties of matrices of zeros and ones.
\newblock {\em Canadian Journal of Mathematics}, 9:371--377, 1957.
\newblock \href {https://doi.org/10.4153/CJM-1957-044-3}
  {\path{doi:10.4153/CJM-1957-044-3}}.

\bibitem{Sherali1990}
H.~D. Sherali and W.~P. Adams.
\newblock A hierarchy of relaxations between the continuous and convex hull
  representations for zero-one programming problems.
\newblock {\em {SIAM} J. Discret. Math.}, 3(3):411--430, 1990.
\newblock \href {https://doi.org/10.1137/0403036} {\path{doi:10.1137/0403036}}.

\bibitem{shor1987class}
Naum~Z Shor.
\newblock Class of global minimum bounds of polynomial functions.
\newblock {\em Cybernetics}, 23(6):731--734, 1987.
\newblock \href {https://doi.org/10.1007/BF01070233}
  {\path{doi:10.1007/BF01070233}}.

\bibitem{Tulsiani09:stoc}
Madhur Tulsiani.
\newblock {{CSP} gaps and reductions in the Lasserre hierarchy}.
\newblock In {\em Proc. 41st Annual {ACM} Symposium on Theory of Computing
  (STOC'09)}, pages 303--312. {ACM}, 2009.
\newblock \href {https://doi.org/10.1145/1536414.1536457}
  {\path{doi:10.1145/1536414.1536457}}.

\bibitem{WZ20}
Marcin Wrochna and Stanislav {\v{Z}}ivn{\'y}.
\newblock Improved hardness for {$H$}-colourings of {$G$}-colourable graphs.
\newblock In {\em Proc. 2020 ACM-SIAM Symposium on Discrete Algorithms
  (SODA'20)}, pages 1426--1435, 2020.
\newblock \href {http://arxiv.org/abs/1907.00872} {\path{arXiv:1907.00872}},
  \href {https://doi.org/10.1137/1.9781611975994.86}
  {\path{doi:10.1137/1.9781611975994.86}}.

\bibitem{MR1392498}
Doron Zeilberger.
\newblock Proof of the alternating sign matrix conjecture.
\newblock {\em Electron. J. Combin.}, 3(2), 1996.
\newblock \href {https://doi.org/10.37236/1271} {\path{doi:10.37236/1271}}.

\bibitem{Zhuk20:jacm}
Dmitriy Zhuk.
\newblock A proof of the {CSP} dichotomy conjecture.
\newblock {\em J. {ACM}}, 67(5):30:1--30:78, 2020.
\newblock \href {http://arxiv.org/abs/1704.01914} {\path{arXiv:1704.01914}},
  \href {https://doi.org/10.1145/3402029} {\path{doi:10.1145/3402029}}.

\end{thebibliography}
}

\end{document}